\documentclass[12pt,draftclsnofoot,onecolumn]{IEEEtran}
\usepackage{pifont}
\usepackage{graphicx}
\usepackage{fancyhdr}
\usepackage{amsmath,amssymb,amstext,amsfonts}
\usepackage{slashbox}
\usepackage{bm}
\usepackage{algorithm}
\usepackage{array}
\usepackage{cite}

\usepackage{subfigure}
\usepackage{balance}
\usepackage{algorithmicx}
\usepackage{algpseudocode}

\usepackage{braket}
\usepackage{tabulary}
\usepackage{diagbox}
\usepackage{amsthm}

\newtheorem{theorem}{Theorem}
\newtheorem{lemma}{Lemma}

\begin{document}

\title{Optimally Displaced Threshold Detection for Discriminating Binary Coherent States Using Imperfect Devices}

\author{Renzhi Yuan, Mufei Zhao, Shuai Han, and Julian Cheng.
\thanks{
Renzhi Yuan and Julian Cheng are with the School of Engineering, The University of British Columbia, Kelowna V1V 1V7, BC, Canada (e-mails: renzhi.yuan@ubc.ca, julian.cheng@ubc.ca). Mufei Zhao and Shuai Han are with the Communication Research Center, Harbin Institute of Technology, Harbin 150080, China (e-mails: 19B905015@stu.hit.edu.cn, hanshuai@hit.edu.cn)
}
}

\maketitle

\begin{abstract}
Because of the potential applications in quantum information processing tasks, discrimination of binary coherent states using generalized Kennedy receiver with maximum \textit{a posteriori} probability (MAP) detection has attracted increasing attentions in recent years. In this paper, we analytically study the performance of the generalized Kennedy receiver having optimally displaced threshold detection (ODTD) in a realistic situation with noises and imperfect devices. We first prove that the MAP detection for a generalized Kennedy receiver is equivalent to a threshold detection in this realistic situation. Then we analyze the properties of the optimum threshold and the optimum displacement for ODTD, and propose a heuristic greedy search algorithm to obtain them. We prove that the ODTD degenerates to the Kennedy receiver with threshold detection when the signal power is large, and we also clarify the connection between the generalized Kennedy receiver with threshold detection and the one-port homodyne detection. Numerical results show that the proposed heuristic greedy search algorithm can obtain a lower and smoother error probability than the existing works.
\end{abstract}


\IEEEpeerreviewmaketitle

\section{Introduction}

The discrimination of binary coherent states plays a crucial role in both classical and quantum information processing tasks, such as the coherent optical communications and quantum key distribution (QKD) \cite{osaki1996derivation,vilnrotter2001quantum,lance2005no,weedbrook2012gaussian,yuan2018free,grosshans2003quantum,lorenz2004continuous,bonato2009feasibility,yuan2020closed,yin2016measurement,ghorai2019asymptotic,yuan2020freespace}. It has been proven that the homodyne detection, which provides the standard quantum limit (SQL), is the best strategy to discriminate binary coherent states when only Gaussian operations and classical communication are allowed \cite{takeoka2008discrimination}. By adopting the non-Gaussian operation devices, e.g., the on/off photodetector (PD) or photon number resolving detector (PNRD), the SQL can be surpassed and the performance of the discrimination is limited by the Helstrom bound \cite{helstrom1969quantum,helstrom1970quantum}. The optimal quantum detection achieving the Helstrom bound can be realized by a Dolinar receiver \cite{dolinar1973optimum}, and it has been experimentally demonstrated \cite{cook2007optical}. However, the Dolinar receiver requires real-time feedback loops and high control complexity. Therefore, near-optimum detections with simple structure have been proposed and studied \cite{kennedy1973near,vilnrotter1984generalization,bondurant1993near,geremia2004distinguishing,takeoka2008discrimination,wittmann2008demonstration,wittmann2010discrimination,wittmann2010demonstration,becerra2015photon,dimario2018robust,dimario2019optimized,shcherbatenko2020sub,yuan2020kennedy}. An important near-optimum quantum receiver is the generalized Kennedy receiver \cite{kennedy1973near,yuan2020freespace}, which consists of displacement operation and on/off PD. However, the generalized Kennedy receiver is vulnerable against noises and device imperfections \cite{vilnrotter1984generalization,takeoka2008discrimination,yuan2020kennedy}. The performance of the generalized Kennedy receiver can be improved by optimizing the displacement operation \cite{takeoka2008discrimination,wittmann2008demonstration,yuan2020kennedy}. Besides, by replacing the on/off PD with PNRD, the generalized Kennedy receiver can achieve robust performance against noises and it attracts more and more attentions in recent years \cite{wittmann2010discrimination,wittmann2010demonstration,becerra2015photon,dimario2018robust,dimario2019optimized,yuan2020kennedy}. This is because the PNRD can provide more information of the received quantum states compared with on/off PD \cite{becerra2015photon,dimario2018robust}.

The idea of combining optimally displaced operation and PNRD for discriminating binary coherent states was first proposed \cite{wittmann2010discrimination} and experimentally demonstrated \cite{wittmann2010demonstration} to improve the performance of intermediate discrimination in a post-selected QKD scheme \cite{lorenz2004continuous}. To enable the robustness of the receiver against noises and device imperfections, the maximum \textit{a posteriori} probability (MAP) criterion was adopted to estimate the input state based on the detected number of photons of PNRD, which extends the discrimination of binary coherent states below the SQL to high input power levels \cite{vilnrotter1984generalization,dimario2018robust,dimario2019optimized,yuan2020kennedy}. The impact of the dark count noise and device imperfections on the generalized Kennedy receiver with MAP detection was simulated and verified by experiments \cite{dimario2018robust}.

However, an analytical study of the impact of noises and devices imperfection for generalized Kennedy receiver with MAP detection has not been performed. Besides, the impact of Gaussian thermal noise on the MAP detection was not considered in these works \cite{wittmann2010discrimination,wittmann2010demonstration,becerra2015photon,dimario2018robust,dimario2019optimized}. The thermal noise is generated by the load resistor of the electric circuit in the receiver, and it can become the dominated noise source compared with the dark count noise when the temperature raises due to long working hours \cite{xu2011impact}. In addition, the communication channel between two legitimate users in a QKD scheme is usually considered as an additive white Gaussian noise (AWGN) channel \cite{ghorai2019asymptotic}, where the additive white Gaussian noise can be equivalently regarded as a thermal noise of the receiver. Therefore, it is meaningful to study the impact of Gaussian thermal noise on the generalized Kennedy receiver with MAP detection for discriminating binary coherent states.

In a previous work \cite{yuan2020kennedy}, we studied the impact of thermal noise on the MAP detection and numerically studied the error probability of the generalized Kennedy receiver when the displacement is optimized after the threshold optimization. However, the device imperfection and other noise sources are not considered. Besides, the error probability given in \cite{yuan2020kennedy} cannot achieve the global optimality when the displacement is optimized after the threshold optimization. In this paper, we extend our previous study to a more realistic situation in the presence of both thermal noise and dark count noise using imperfect devices. We prove that the MAP detection for generalized Kennedy receiver in this realistic situation is equivalent to a threshold detection. We call the generalized Kennedy receiver with threshold detection adopting optimum displacement the optimally displaced threshold detection (ODTD). The main contributions of this work include:
\begin{itemize}
\item We prove that the MAP detection for generalized Kennedy receiver is equivalent to a threshold detection in the presence of both thermal noise and dark count noise using imperfect devices (Theorem \ref{theorem 1}).
\item We prove that the ODTD degenerates to the Kennedy receiver with threshold detection when the signal power is large (Theorem \ref{theorem 2}).
\item We clarify the connection between the generalized Kennedy receiver with threshold detection and the one-port homodyne detection (Theorem \ref{theorem 3}).
\item We propose a heuristic greedy search method to search the optimum threshold and the optimum displacement for ODTD, and obtain a lower and smoother error probability than the existing works.
\end{itemize}

The rest of this paper is organized as follows. In Section \ref{sect:KennedyReceiverWithThresholdDetection}, we establish the MAP detection model. Section \ref{sect:KennedyReceiverWithOptimalThresholdAndDisplacement} discusses optimum threshold and optimum displacement of the ODTD. Section \ref{sect:NumericalResults} presents some numerical results, and a brief conclusion is summarized in Section \ref{sect:Conclusion}.

\section{Generalized Kennedy Receiver With Threshold Detection}
\label{sect:KennedyReceiverWithThresholdDetection}

\begin{figure}
\begin{center}
\includegraphics[width=0.9\textwidth, draft=false]{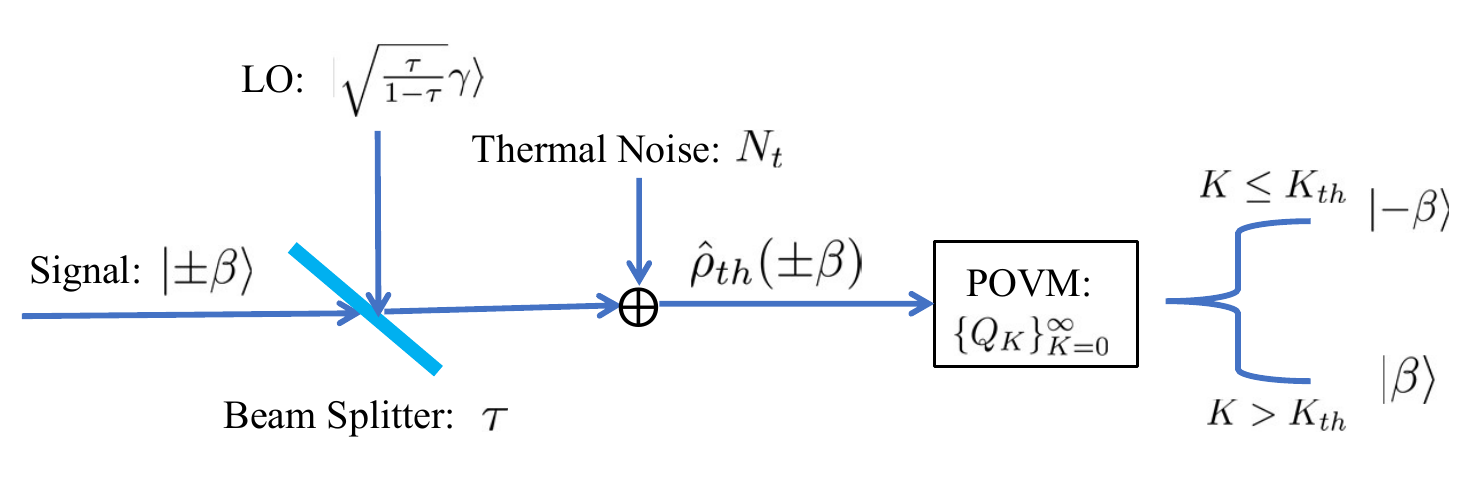}
\caption{Threshold detection for discriminating displaced binary coherent states}
\vspace{-0.4cm}
\label{Fig:Threshold_Detection}
\end{center}
\end{figure}

The configuration of threshold detection for discriminating displaced binary coherent states is shown in Fig. \ref{Fig:Threshold_Detection}. The input coherent state $\ket{\pm \beta}$ is first displaced by a displacement operator $\hat{D}(\gamma)$. The displacement operator $\hat{D}(\gamma)$ is achieved by combining the input signal with a local oscillator (LO) $\ket{\beta_{LO}}$ using a beam splitter with transmittance rate $\tau \to 1$. Here we set the nominal value of the LO as $\beta_{LO}=\sqrt{\frac{\tau}{1-\tau}}\gamma$ to null out $\ket{-\beta}$ when $\gamma=\beta$. Then the displaced coherent state is contaminated by the thermal noise, which is represented by the number of thermal photons $N_t$. The thermal noise contaminated state $\hat{\rho}_{th}(\pm \beta)$ is measured using a PNRD. The PNRD can be characterized by a set of positive-operator valued measure (POVM) operators $\{\hat{Q}_K\}_{K=0}^{\infty}$, where $\hat{Q}_K$ is the measurement operator corresponding to $K$ detected number of photons. The outcome of the detection is decided as $\ket{-\beta}$ when $K \leq K_{th}$ or $\ket{\beta}$ when $K > K_{th}$, where $K_{th}$ is a given threshold.

Because the output of the PNRD is the number of photons, it is convenient to establish our analysis in the Fock space, which is the Hilbert space spanned by a set of number states $\{\ket{n}, n=0,1,2,\cdots\}$ \cite{mandel1995optical}. Next we first present a brief review of coherent states and $P$-representation of any quantum state in Fock space. Then we use the $P$-representation to derive the probability of detecting $K$ photons of PNRD, and introduce the MAP detection for the displaced binary coherent states. In the last part of this section, we prove that the MAP detection is equivalent to a threshold detection.

\subsection{Coherent States And $P$-representation}
Due to the maturity of laser techniques, coherent states are usually employed in both classical coherent communications and quantum communications \cite{mandel1995optical,wang2007quantum,weedbrook2012gaussian}. A coherent state $\ket{\alpha}$, where $\alpha \in \mathbb{C}$ and $\mathbb{C}$ is the set of complex numbers, in this Fock space is represented by a superposition of the number states as \cite{glauber1963quantum,glauber1963coherent}
\begin{equation}
\label{Equa:CoherentState}
{\ket{\alpha}} = \sum_{n=0}^{\infty}e^{-\frac{1}{2}{|\alpha|}^2}\frac{{\alpha}^n}{\sqrt{n!}}\ket{n}.
\end{equation}
\noindent All the coherent states ${\{\ket{\alpha}, \alpha \in \mathbb{C}\}}$ form an overcomplete basis of the Hilbert space. Therefore, a density operator ${\hat{\rho}}$ of this Hilbert space can be decomposed by coherent states as
\begin{equation}
\label{Equa:DensityfromP}
\hat{\rho} = \int_{\alpha}P(\alpha)\ket{\alpha}\bra{\alpha}\mathrm{d}^2\alpha
\end{equation}
\noindent where $\mathrm{d}^2\alpha=\mathrm{d}\Re(\alpha)\mathrm{d}\Im(\alpha)$. ${P(\alpha)}$ is the $P$-function of the density operator $\hat{\rho}$ and this representation of density operator is called the $P$-representation \cite{glauber1963quantum,glauber1963coherent}.

\subsection{Probability Of Detecting $K$ Photons }
After passing through the beam splitter with transmission rate $\tau$, the coherent state $\ket{\pm \beta}$ is displaced as $\ket{\sqrt{\tau}(\pm \beta+\gamma)}$. Then the displaced coherent state is contaminated by thermal noise. Using the $P$-representation, the density operator of this thermal noise contaminated quantum state can be obtained as \cite{glauber1963quantum,glauber1963coherent}
\begin{equation}
\label{Rho_th}
\hat{\rho}_{th}(\pm \beta,\gamma)=\int_{\mathbb{C}}\frac{1}{\pi N_t}e^{-\frac{|\alpha-\sqrt{\tau}(\pm \beta+\gamma)|^2}{N_t}}\ket{\alpha}\bra{\alpha}\mathrm{d}^2\alpha
\end{equation}
\noindent where $N_t$ is the equivalent average number of photons generated by the thermal noise, and it is also called the ``thermal photons".

This thermal noise contaminated quantum state $\hat{\rho}_{th}(\pm \beta,\gamma)$ is measured by a set of POVM operators $\{\hat{Q}_K\}_{K=0}^{\infty}$, where the measurement operator $\hat{Q}_K$ can be obtained as
\begin{equation}
\label{Q_K_1}
\hat{Q}_K =\sum_{m=K}^{\infty}\binom{m}{K}\eta^i(1-\eta)^{m-K}\ket{m}\bra{m}
\end{equation}
\noindent where $\eta$ is the quantum efficiency of the PNRD. Specially, when $\eta=1$, the measurement operator $\hat{Q}_K$ degenerates to $\hat{Q}_K=\ket{K}\bra{K}$.

Then the probability of detecting $K$ photons when state $\ket{\pm \beta}$ is transmitted can be obtained by
\begin{equation}
\label{P_K}
P(K|\pm \beta,\gamma)=\text{Tr}\left(\hat{Q}_K\hat{\rho}_{th}(\pm \beta,\gamma)\right)
\end{equation}
\noindent where $\text{Tr}(\cdot)$ is the trace operation.

Substituting \eqref{Rho_th} and \eqref{Q_K_1} into \eqref{P_K}, we can obtain
\begin{equation}
\label{P_K_0}
\begin{aligned}
P(K|\pm \beta, \gamma)&=\sum_{m=K}^{\infty}\binom{m}{K}\eta^K(1-\eta)^{m-K} \left[\int_{\mathbb{C}}\frac{1}{\pi N_t}e^{-\frac{|\alpha-\sqrt{\tau}(\pm \beta+\gamma)|^2}{N_t}-|\alpha|^2}\frac{|\alpha|^{2m}}{m!}\mathrm{d}^2\alpha\right]\\
&=\sum_{m=K}^{\infty}\binom{m}{K}\eta^K(1-\eta)^{m-K}\frac{N_t^m}{(N_t+1)^{m+1}}e^{-\frac{\left \langle n \right \rangle_{\pm}}{N_t+1}}L_m\left(-\frac{\left \langle n \right \rangle_{\pm}}{N_t(N_t+1)}\right)\\
\end{aligned}
\end{equation}
\noindent where $L_m(x)$ is the Laguerre polynomial of order $m$; and $\left \langle n \right \rangle_{\pm}$ is the average number of photons of the displaced coherent state $\ket{\sqrt{\tau}(\pm \beta+\gamma)}$. When the imperfections of the displacement, including the phase noise and the mismatch between the signal and the LO, are considered, $\left \langle n \right \rangle_{\pm}$ can be approximated by \cite{becerra2015photon,dimario2018robust}
\begin{equation}
\label{n_pm}
\left \langle n \right \rangle_{\pm}= \tau (|\beta|^2+|\gamma|^2 \pm 2 \xi |\beta||\gamma|)
\end{equation}
\noindent where $\xi \in [0,1]$ is the interference visibility, which can be obtained from the interference measurement. For perfect devices and interference, $\xi=1$. Besides, if the dark count noise is also considered, $\left \langle n \right \rangle_{\pm}$ needs to be replaced by $\left \langle n' \right \rangle_{\pm} = \left \langle n \right \rangle_{\pm}+\nu$, where $\nu$ is the average number of photons due to the dark count noise.

Using the identity $\sum_{m=K}^{\infty}\binom{m}{K}(1-\frac{1}{t})^{m-K} L_m(x)=t^{K+1}e^{-(t-1)x}L_K(xt)$ and replace $\left \langle n \right \rangle_{\pm}$ with $\left \langle n' \right \rangle_{\pm}$, we can further rewrite \eqref{P_K_0} as
\begin{equation}
\label{P_K_1}
\begin{aligned}
P(K|\pm \beta, \gamma)=\frac{(\eta N_t)^K}{(\eta N_t+1)^{K+1}} e^{-\frac{\left \langle n' \right \rangle_{\pm}}{N_t+1/\eta}}L_K\left(-\frac{\left \langle n' \right \rangle_{\pm}}{N_t(\eta N_t+1)}\right).
\end{aligned}
\end{equation}

Based on the probability of detecting $K$ photons, we can now introduce the MAP criteria to estimate the transmitted state, which is based on the test
\begin{equation}
\label{MAP test}
p_0 P(K|- \beta, \gamma)\mathop{\gtreqless} \limits_{\ket{\beta}}^{\ket{-\beta}} p_1 P(K|\beta, \gamma)
\end{equation}
\noindent where $p_0$ and $p_1$ are the prior probabilities of transmitting $\ket{-\beta}$ and $\ket{\beta}$, respectively.

\subsection{Threshold Detection}
\label{MAPandThresholdDetection}
In our previous work \cite{yuan2020kennedy}, we proved that the MAP detection \eqref{MAP test} is equivalent to a threshold detection when the displacement $\gamma=\beta$ and only thermal noise is considered. Here we generalize this result to a more realistic situation with arbitrary displacement $\gamma$, and both thermal noise and dark count noise are considered using imperfect device.

\begin{theorem}\label{theorem 1}
When both thermal noise and dark count noise are considered, the MAP detection \eqref{MAP test} using imperfect device is equivalent to a threshold detection
\begin{equation}
\label{Threshold test}
K \mathop{\lesseqgtr} \limits_{\ket{\beta}}^{\ket{-\beta}} K_{th}.
\end{equation}
\end{theorem}

\begin{proof}
The MAP test in \eqref{MAP test} is equivalent to
\begin{equation}
\label{MAP test_1}
\frac{P(K|\beta, \gamma)}{P(K|-\beta, \gamma)}\mathop{\lesseqgtr} \limits_{\ket{\beta}}^{\ket{-\beta}} \frac{p_0}{p_1}.
\end{equation}

Let $g(K)\triangleq P(K|\beta, \gamma)/P(K|-\beta, \gamma)$. If $g(K)$ is an increasing function of $K$, then there exists an integer $K_{th}$: for any $K\leq K_{th}$, $g(K)\leq p_0/p_1$ and state $\ket{-\beta}$ is selected; for any $K>K_{th}$, $g(K)>p_0/p_1$ and state $\ket{\beta}$ is selected. Clearly this is the threshold test in \eqref{Threshold test}. Therefore, the key is to prove that $g(K)$ is an increasing function of $K$, which is equivalent to prove that
\begin{equation}
\label{F_K}
\frac{P(K+1|\beta, \gamma)}{P(K|\beta, \gamma)} \geq \frac{P(K+1|-\beta, \gamma)}{P(K|-\beta, \gamma)}
\end{equation}
\noindent for any $K \geq 0$.

By substituting \eqref{P_K_1} into \eqref{F_K} and canceling the same terms on both sides, we can obtain
\begin{equation}
\begin{aligned}
\frac{L_{K+1}\left(-\frac{\left \langle n' \right \rangle_{+}}{N_t(\eta N_t+1)}\right)}{L_{K}\left(-\frac{\left \langle n' \right \rangle_{+}}{N_t(\eta N_t+1)}\right)} \geq \frac{L_{K+1}\left(-\frac{\left \langle n' \right \rangle_{-}}{N_t(\eta N_t+1)}\right)}{L_{K}\left(-\frac{\left \langle n' \right \rangle_{-}}{N_t(\eta N_t+1)}\right)}.
\end{aligned}
\end{equation}

Because $-\frac{\left \langle n' \right \rangle_{+}}{N_t(\eta N_t+1)} \leq -\frac{\left \langle n' \right \rangle_{-}}{N_t(\eta N_t+1)}$, this indicates that $f(x)\triangleq \frac{L_{K+1}(x)}{L_K(x)}$ needs to be a decreasing function of $x$ for any $x\leq 0$ and $K \geq 0$,  which is directly followed from Lemma \ref{lemma_2} in Appendix \ref{Appendix A}. Therefore, $g(K)$ is an increasing function of $K$; and the MAP detection in \eqref{MAP test} is equivalent to the threshold detection in \eqref{Threshold test}.
\end{proof}

Theorem \ref{theorem 1} indicates that instead of calculating the posteriori probability, we can calculate the threshold $K_{th}$ in advance and then decide the outcome based on the threshold test in \eqref{Threshold test}. This can simplify the design of the receiver.

\section{Optimally Displaced Threshold Detection}
\label{sect:KennedyReceiverWithOptimalThresholdAndDisplacement}
We have introduced the threshold detection for discriminating displaced binary coherent states. The error probability of the discrimination given threshold $K_{th}$ and displacement $\gamma$ is $P_{e}(K_{th},\gamma)=p_0 \text{Pr}(\left.K>K_{th}\right|\ket{-\beta})+p_1\text{Pr}(\left.K\leq K_{th}\right|\ket{\beta})$, which can be obtained as
\begin{equation}
\label{Pe_0}
\begin{aligned}
P_{e}(K_{th},\gamma)&=p_0 \sum_{K=K_{th}+1}^{\infty} P(K|-\beta, \gamma)+ p_1 \sum_{K=0}^{K_{th}} P(K|\beta, \gamma)\\
&=p_0+p_1 \sum_{K=0}^{K_{th}} P(K|\beta, \gamma)-p_0\sum_{K=0}^{K_{th}} P(K|-\beta, \gamma).
\end{aligned}
\end{equation}
\noindent Without loss of generality, we suppose $\beta$ and $\gamma$ are real numbers. Then by substituting \eqref{P_K_1} and \eqref{n_pm} into \eqref{Pe_0}, we can obtain an explicit expression for the error probability $P_e(K_{th},\gamma)$ as
\begin{equation}
\label{Pe_1}
\begin{aligned}
P_{e}(K_{th},\gamma)&=p_0+p_1 e^{-\frac{\tau (\beta^2+\gamma^2 + 2 \xi \beta\gamma)+\nu}{N_t+1/\eta}} \sum_{K=0}^{K_{th}}\frac{(\eta N_t)^K}{(\eta N_t+1)^{K+1}} L_K\left(-\frac{\tau (\beta^2+\gamma^2 + 2 \xi \beta\gamma)+\nu}{N_t(\eta N_t+1)}\right)\\
&\quad \quad -p_0 e^{-\frac{\tau (\beta^2+\gamma^2 - 2 \xi \beta\gamma)+\nu}{N_t+1/\eta}} \sum_{K=0}^{K_{th}} \frac{(\eta N_t)^K}{(\eta N_t+1)^{K+1}} L_K\left(-\frac{\tau (\beta^2+\gamma^2 - 2 \xi \beta\gamma)+\nu}{N_t(\eta N_t+1)}\right).
\end{aligned}
\end{equation}

The ODTD adopts the optimum threshold $K_{th}^*$ and the optimum displacement $\gamma^*$ to achieve the minimum error probability of discriminating binary coherent states. Then the optimum threshold $K_{th}^*$ and optimum displacement $\gamma^*$ for minimum error probability can be obtained by solving the optimization problem
\begin{equation}
\label{OptimizationForThresholdAndDisplacement}
(K_{th}^*,\gamma^*)=\mathop{\arg\min}_{(K_{th},\gamma)} P_e(K_{th},\gamma).
\end{equation}

However, it is challenging to solve \eqref{OptimizationForThresholdAndDisplacement} analytically. Before going further, we first discuss a simple case with $\gamma=\beta$.

\subsection{Kennedy Receiver With Threshold Detection ($\gamma=\beta$)}
\label{KennedyReceiverWithThresholdDetection}
When $\gamma=\beta$, the ODTD becomes a Kennedy receiver with threshold detection \cite{vilnrotter1984generalization,yuan2020kennedy}. Then we only need to optimize the threshold $K_{th}$.

Substituting $\gamma=\beta$ into \eqref{n_pm}, we can obtain
\begin{equation}\label{n_pm_1}
\left \langle n' \right \rangle_{\pm}=2\tau N_s(1\pm \xi)+\nu
\end{equation}
\noindent where $N_s=|\beta|^2$ is the average number of photons per bit of the transmitted signal, which is also called the ``signal photons".

Substituting \eqref{P_K_1} and \eqref{n_pm_1} into \eqref{Pe_0}, we can obtain the error probability $P_{e}(K_{th})$ of Kennedy receiver with threshold detection as
\begin{equation}
\label{Pe_Threshold_0}
\begin{aligned}
P_{e}(K_{th})&=p_0+p_1 e^{-\frac{2\tau N_s(1+\xi)+\nu}{N_t+1/\eta}} \sum_{K=0}^{K_{th}}\frac{(\eta N_t)^K}{(\eta N_t+1)^{K+1}} L_K\left(-\frac{2\tau N_s(1+\xi)+\nu}{N_t(\eta N_t+1)}\right)\\
&\quad \quad -p_0 e^{-\frac{2\tau N_s (1-\xi)+\nu}{N_t+1/\eta}} \sum_{K=0}^{K_{th}} \frac{(\eta N_t)^K}{(\eta N_t+1)^{K+1}} L_K\left(-\frac{2\tau N_s (1-\xi)+\nu}{N_t(\eta N_t+1)}\right).
\end{aligned}
\end{equation}

\subsubsection{Optimum Threshold $K_{th}^*$}
From \eqref{Pe_Threshold_0} we can observe that when $K_{th}$ increases, the second term of the error probability increases and the third term decreases. Therefore, after reaching the optimum threshold $K_{th}^*$, when the threshold increases from $K_{th}$ to $K_{th}+1$, the increment of the second term must be greater than or equal to the decrement of the third term. This means that the optimum threshold $K_{th}^*$ is the minimum $K_{th}$ which satisfies the following inequality
\begin{equation}
\label{Optimal_Threshold_0}
\begin{aligned}
p_1 e^{-\frac{2\tau N_s(1+\xi)+\nu}{N_t+1/\eta}}L_{K_{th}+1}\left(-\frac{2\tau N_s(1+\xi)+\nu}{N_t(\eta N_t+1)}\right)
\geq
p_0  e^{-\frac{2\tau N_s(1-\xi)+\nu}{N_t+1/\eta}}L_{K_{th}+1}\left(-\frac{2\tau N_s(1-\xi)+\nu}{N_t(\eta N_t+1)}\right).
\end{aligned}
\end{equation}
\noindent By canceling the same terms on both sides, we can rewrite this inequality as
\begin{equation}
\label{Optimal_Threshold_1}
\begin{aligned}
\frac{L_{K_{th}+1}\left(-\frac{2\tau N_s(1+\xi)+\nu}{N_t(\eta N_t+1)}\right)}{L_{K_{th}+1}\left(-\frac{2\tau N_s(1-\xi)+\nu}{N_t(\eta N_t+1)}\right)} \geq \frac{p_0}{p_1} e^{\frac{4\tau \xi N_s}{N_t+1/\eta}}.
\end{aligned}
\end{equation}
\noindent Then the optimum threshold $K_{th}^*$ can be obtained by solving the following optimization problem
\begin{equation}
\label{OptimizationForThreshold_0}
\begin{aligned}
K_{th}^*=\mathop{\min} &\quad K_{th}\\
\text{s.t.} &\quad \frac{L_{K_{th}+1}\left(-\frac{2\tau N_s(1+\xi)+\nu}{N_t(\eta N_t+1)}\right)}{L_{K_{th}+1}\left(-\frac{2\tau N_s(1-\xi)+\nu}{N_t(\eta N_t+1)}\right)} \geq \frac{p_0}{p_1} e^{\frac{4\tau \xi N_s}{N_t+1/\eta}}.
\end{aligned}
\end{equation}
\noindent Compared with the optimization problem \eqref{OptimizationForThresholdAndDisplacement}, this optimization problem \eqref{OptimizationForThreshold_0} is much simpler.

\subsubsection{Lower Bound Of $P_{e}(K_{th})$}
\label{lowerbounds}
Because the interference visibility $\xi$ is usually smaller than 1, the last two terms of \eqref{Pe_Threshold_0} diminishes to zero as $N_s$ approaches $\infty$. Then we have
\begin{equation}\label{Lowerbound0}
\lim \limits_{N_s \to \infty} P_e(K_{th})= p_0
\end{equation}
\noindent for $\xi<1$. This indicates that the imperfect interference visibility greatly degrades the performance of the receiver when the signal power is large.

The imperfect interference visibility is mainly due to the phase noise and the mismatch between the signal and the LO. By introducing phase-locked loops into the receiver, the interference visibility can be greatly improved. In the reported works \cite{becerra2015photon,dimario2018robust}, the interference visibility can achieve $\xi=99.8\%$, which makes it possible for discriminating binary coherent states below SQL with high signal power. If the interference visibility $\xi \to 1$, then the error probability $P_e(K_{th})$ becomes
\begin{equation}
\label{Pe_Threshold_1}
\begin{aligned}
P_{e}(K_{th})|_{\xi \to 1}&=p_0 \left(\frac{\eta N_t}{\eta N_t+1}\right)^{K_{th}+1}+p_1 e^{-\frac{4\tau N_s+\nu}{N_t+1/\eta}}\sum_{K=0}^{K_{th}}\frac{(\eta N_t)^K}{(\eta N_t+1)^{K+1}} L_K\left(-\frac{4\tau N_s+\nu}{N_t(\eta N_t+1)}\right).
\end{aligned}
\end{equation}
\noindent When $N_s \to \infty$, the second term of \eqref{Pe_Threshold_1} diminishes to zero and the error probability has a lower bound
\begin{equation}\label{Lowerbound1}
\lim \limits_{N_s \to \infty} P_e(K_{th})|_{\xi \to 1}= p_0 \left(\frac{\eta N_t}{\eta N_t+1}\right)^{K_{th}+1}.
\end{equation}
\noindent Because we cannot achieve a perfect interference visibility, $P_{e}(K_{th})\geq P_{e}(K_{th})|_{\xi \to 1}$ always holds. Then \eqref{Lowerbound1} is also a lower bound for $P_{e}(K_{th})$.

The original Kennedy receiver (with on/off detection) can be obtained by setting $K_{th}=0$. Therefore, from \eqref{Lowerbound1} we can observe that the lower bound of error probability for the threshold detection with threshold $K_{th}$ has a gain of $\left(\frac{\eta N_t}{\eta N_t+1}\right)^{K_{th}}$ compared with the lower bound of error probability for the original Kennedy receiver. When the thermal photons $N_t$ is much smaller than 1, this gain approaches $(\eta N_t)^{K_{th}}$. The lower bound \eqref{Lowerbound1} directly shows how the threshold limits the performance of the receiver under large signal power.

Obviously, the lower bound given in \eqref{Lowerbound1} is not a satisfactory bound because it only bounds the error probability for a given threshold $K_{th}$. From the condition \eqref{Optimal_Threshold_1}, we know that the optimum threshold $K_{th}^*$ varies with the signal photons $N_s$, thus the bound in \eqref{Lowerbound1} varies with signal photons, too. Therefore, a more meaningful lower bound should be irrelevant to $K_{th}$, which can be approximated by the envelope of a set of curves $\{P_{e}(K_{th}), K_{th}=0,1,2,\cdots\}$. However, because $K_{th}$ is a discrete variable, the envelope of these curves is undefined. Then the derivation of an analytical tight lower bound for the error probability is still an open question. In Section \ref{sect:NumericalResults}, we will use a curve-fitting method to find a practical asymptotic lower bound for the error probability.

\subsection{Optimally Displaced Threshold Detection ($\gamma \neq \beta$)}\label{Optimizations}
The Kennedy receiver can beat the SQL when the average number of signal photons $N_s>0.4$ \cite{takeoka2008discrimination}. By optimizing the displacement of Kennedy receiver with on/off photodetector, the optimized displacement receiver (ODR) can beat the SQL under all signal photons \cite{takeoka2008discrimination,wittmann2008demonstration}. Inspired by this, in our previous work \cite{yuan2020kennedy}, we studied the performance of the discrimination when the displacement is further optimized after the optimization of threshold. However, this may not achieve the global minimum error probability of the discrimination. A more rigorous method is to solve the two-variable optimization problem in \eqref{OptimizationForThresholdAndDisplacement}.

As we have mentioned, it is challenging to solve \eqref{OptimizationForThresholdAndDisplacement} analytically, and even the numerical solution for \eqref{OptimizationForThresholdAndDisplacement} is not trivial. Because the threshold $K_{th}$ is a discrete variable, the ordinary gradient descent algorithm cannot be directly applied. An intuitive idea is to use the coordinate descent algorithm. Coordinate descent successively minimizes one single coordinate to find the minimum of a function. In our case, the algorithm minimize one variable, either $K_{th}$ or $\gamma$, while fixing the other one at each iteration.

\subsubsection{Optimize $K_{th}$ Given $\gamma$}
For a fixing displacement $\gamma$, the optimum threshold $K_{th}^*(\gamma)$ can be obtained by minimizing the error probability in \eqref{Pe_1} over $K_{th}$. From \eqref{Pe_1} we can observe that as $K_{th}$ increases, the second term of the error probability increases and the third term decreases. Therefore, similar to \eqref{OptimizationForThreshold_0}, we can obtain the optimum threshold $K_{th}^*(\gamma)$ by solving the following optimization problem:
\begin{equation}
\label{OptimizationForThreshold_1}
\begin{aligned}
K_{th}^*(\gamma)=\mathop{\min} &\quad K_{th}\\
\text{s.t.} &\quad \frac{L_{K_{th}+1}\left(-\frac{\tau (\beta^2+\gamma^2 + 2 \xi \beta\gamma)+\nu}{N_t(\eta N_t+1)}\right)}{L_{K_{th}+1}\left(-\frac{\tau (\beta^2+\gamma^2 - 2 \xi \beta\gamma)+\nu}{N_t(\eta N_t+1)}\right)} \geq \frac{p_0}{p_1} e^{\frac{4\tau \xi \beta \gamma}{N_t+1/\eta}}.
\end{aligned}
\end{equation}
\noindent When $\gamma=\beta$, this optimization problem \eqref{OptimizationForThreshold_1} degenerates to the optimization problem \eqref{OptimizationForThreshold_0}.

\subsubsection{Optimize $\gamma$ Given $K_{th}$}
For a fixing threshold $K_{th}$, the optimum displacement $\gamma^*$ can be obtained by letting the partial derivative $\frac{\partial P_{e}(K_{th},\gamma)}{\partial \gamma}$ equal zero. After some algebra, one can find the optimum displacement $\gamma^*$ satisfying
\begin{equation}
\label{OptimizationForDisplacement_0}
\begin{aligned}
\frac{\gamma^*+\xi\beta}{\gamma^*-\xi\beta}=\frac{p_0}{p_1}e^{\frac{4\tau\xi\beta\gamma^*}{N_t+1/\eta}}
\frac{\sum_{K=0}^{K_{th}}\frac{(\eta N_t)^K}{(\eta N_t+1)^{K+1}}h(\gamma^*;K,-\beta)}{\sum_{K=0}^{K_{th}}\frac{(\eta N_t)^K}{(\eta N_t+1)^{K+1}}h(\gamma^*;K,\beta)}
\end{aligned}
\end{equation}
\noindent where $h(\gamma^*;K,\pm \beta)$ is defined as
\begin{equation}\label{h_K}
h(\gamma^*;K,\pm \beta)\triangleq
\left\{
\begin{array}{ll}
1, \text{ if } K=0\\
L_{K}\left(-\frac{\tau (\beta^2+\gamma^2 \pm 2 \xi \beta\gamma)+\nu}{N_t(\eta N_t+1)}\right)-\frac{1}{\eta N_t}L_{K-1}^1 \left(-\frac{\tau (\beta^2+\gamma^2 \pm 2 \xi \beta\gamma)+\nu}{N_t(\eta N_t+1)}\right), \text{ if } K>0
\end{array}
\right.
\end{equation}
\noindent where $L_m^{\alpha}(x)$ is the generalized Laguerre polynomial of order $m$ with parameter $\alpha$. Eq. \eqref{OptimizationForDisplacement_0} degenerates to the condition for ODR given in \cite[eq. (23)]{takeoka2008discrimination} when $K_{th}=0$, $\tau=\xi=1$, $p_0=p_1=0.5$, and $N_t=0$.

From \eqref{h_K}, we can see that the computational complexity of solving \eqref{OptimizationForDisplacement_0} increases greatly as $K_{th}$ increases when $K_{th}>0$. Therefore, in practical implementation, it is more efficient to solve the single-variable optimization problem $\gamma^*(K_{th})=\mathop{\arg\min}_{\gamma} P_e(K_{th},\gamma)$ using the line-search algorithm compared with solving equation \eqref{OptimizationForDisplacement_0} when $K_{th}$ is large.

The coordinate descent searches the optimum solution to \eqref{OptimizationForThresholdAndDisplacement} by successively optimizing $K_{th}$ through solving \eqref{OptimizationForThreshold_1} and optimizing $\gamma$ through solving \eqref{OptimizationForDisplacement_0}. However, because $K_{th}$ is a discrete variable, the update of the threshold $K_{th}$ by solving \eqref{OptimizationForThreshold_1} can fail when the variation of two successive $\gamma$ is too small. Then the coordinate descent algorithm will be trapped at these points. To solve this problem, we propose a heuristic greedy search algorithm.

\subsubsection{Heuristic Greedy Search}\label{algorithm}
The optimization of $Pe(K_{th},\gamma)$ is equivalent to the optimization of the single-variable function $P_e(K_{th},\gamma^*(K_{th}))$. Then the objective optimization problem in \eqref{OptimizationForThresholdAndDisplacement} is equivalent to the following single-variable optimization problem
\begin{equation}
\label{OptimizationForThresholdAndDisplacement_1}
K_{th}^*=\mathop{\arg\min}_{K_{th}} P_e(K_{th},\gamma^*(K_{th})).
\end{equation}

Because $K_{th}$ ranges from 0 to $\infty$, we cannot use the brute-force search to solve \eqref{OptimizationForThresholdAndDisplacement_1}. However, if the discrete function $P_e(K_{th},\gamma^*(K_{th}))$ have the following convex property
\begin{equation}\label{Relation}
\left\{
\begin{array}{ll}
P_e(K_{th},\gamma^*(K_{th}))>P_e(K_{th}+1,\gamma^*(K_{th})) \text{ for } K_{th}<K_{th}^*\\
P_e(K_{th},\gamma^*(K_{th}))<P_e(K_{th}+1,\gamma^*(K_{th})) \text{ for } K_{th}\leq K_{th}^*
\end{array}
\right.
\end{equation}
\noindent then the local optimality of $P_e(K_{th},\gamma^*(K_{th}))$ guarantees the global optimality \cite{murota2001relationship}. Then we can use a heuristic greedy search to solve the optimization problem in \eqref{OptimizationForThresholdAndDisplacement_1}.

\renewcommand{\algorithmicrequire}{\textbf{Input:}}
\renewcommand{\algorithmicensure}{\textbf{Output:}}
\begin{algorithm}
\caption{Heuristic Greedy Search }
        \begin{algorithmic}[1]
        \State Initialization: $i \gets 1$, $K_{th,i} \gets K_{th}^*(\xi\beta)$,$\gamma_i \gets \gamma^*(K_{th,i})$, $P_{e} \gets 1$, $d \gets 1$
        \If {$P_e(K_{th,i},\gamma_i)<P_e(K_{th,i}+1,\gamma^*(K_{th,i}+1))$}
        \State $d \gets -1$
        \EndIf
        \While {$(P_e-P_{e}(K_{th,i},\gamma_i))/P_e > \epsilon$ \textbf{and} $i \leq M$  \textbf{and} $K_{th,i} \geq 0$}
            \State $P_{e} \gets P_{e}(K_{th,i},\gamma_{i})$
            \State $K_{th,i+1} \gets K_{th,i}+d$
            \State $\gamma_{i+1} \gets \gamma^*(K_{th,i+1})$
            \State $i \gets i+1$
        \EndWhile
        \State $K_{th}^* \gets K_{th,i-1}, \gamma^* \gets \gamma_{i-1}, P_e \gets P_e(K_{th}^*,\gamma^*)$
        \end{algorithmic}\label{Heuristic Searching}
\end{algorithm}

The pseudocode of heuristic greedy search is summarized in Algorithm \ref{Heuristic Searching}, where $\gamma_i$ and $K_{th,i}$ are the displacement and threshold for the $i$th iteration, respectively; $\epsilon$ is the required relative error; $M$ is the maximum number of iterations; $d$ is searching direction for the threshold. Line 1 is the initialization of all variables. Lines 2-4 decide the searching direction of the threshold. Line 5 defines the halt condition. Lines 6-9 update all variables. Line 11 presents the output of the algorithm.

Specially, we set the initial values for the threshold and displacement as $K_{th,1}=K_{th}^*(\xi\beta)$ and $\gamma_1=\gamma^*(K_{th,1})$, respectively. This is based on the result of the following Theorem \ref{theorem 2}.

\begin{theorem}\label{theorem 2}
For any threshold $K_{th} \geq 0$, the optimum displacement $\gamma^*$ approaches $\xi\beta$ when $\beta$ approaches $\infty$.
\end{theorem}

\begin{proof}
We have shown that the optimum displacement satisfies \eqref{OptimizationForDisplacement_0}. Then we can rewrite \eqref{OptimizationForDisplacement_0} as
\begin{equation}
\label{Gamma-Beta_2}
\gamma^*-\xi\beta= e^{-\frac{4\tau\xi\beta\gamma^*}{N_t+1/\eta}}\frac{p_1(\gamma^*+\xi\beta)}{p_0} \frac{\sum_{K=0}^{K_{th}}\frac{(\eta N_t)^K}{(\eta N_t+1)^{K+1}}h(\gamma^*;K,\beta)}{\sum_{K=0}^{K_{th}}\frac{(\eta N_t)^K}{(\eta N_t+1)^{K+1}}h(\gamma^*;K,-\beta)}.
\end{equation}
\noindent When $\beta$ approaches $\infty$, the right side of \eqref{Gamma-Beta_2} approaches zero due to the presence of exponential term $e^{-\frac{4\tau\xi\beta\gamma^*}{N_t+1/\eta}}$. Therefore, we have $\gamma^* \to \xi\beta$.
\end{proof}

Theorem \ref{theorem 2} indicates that when the number of signal photons $N_s=|\beta|^2$ is large, the optimum displacement $\gamma^*$ degenerates to $\xi\beta$. Therefore, when the required precision is relatively low, we can let $\gamma^*=\xi\beta$ and optimize the threshold $K_{th}$ only. Besides, if $\xi$ approaches 1, then we have $\gamma^* \approx \beta$ when $\beta$ is large, thus the ODTD degenerates to the Kennedy receiver with threshold detection in \ref{KennedyReceiverWithThresholdDetection}.

Moreover, Theorem \ref{theorem 2} also implies that the optimum displacement is near to $\xi \beta$ when $\beta$ is large. Therefore, $(K_{th}^*(\xi\beta), \gamma^*(K_{th}^*(\xi\beta))$ can be a good initial point for $(K_{th},\gamma)$ of the heuristic greedy search algorithm.

\subsection{When $\gamma$ Approaches $\infty$}
Theorem \ref{theorem 2} presents a special case when the signal strength $N_s$ approaches $\infty$. In this subsection, we discuss another special case when LO strength approaches $\infty$.

The configuration given in Fig. \ref{Fig:Threshold_Detection} is similar to the one-port homodyne detection \cite{yuen1983noise,schumaker1984noise}. Especially, transmission rate of the beam splitters in both the ODTD and the one-port homodyne detection are high, i.e., $\tau \to 1$. Therefore, it is necessary to clarify the relation between the ODTD and the one-port homodyne detection.

The main difference between ODTD and one-port homodyne detection is the strength of the LO. The LO $\ket{\beta_{LO}}$ in one-port homodyne detection needs to be high enough, i.e., $\beta_{LO} \to \infty$. However, the LO in ODTD is designed as a specifical value $\ket{\beta_{LO}}=\ket{\sqrt{\frac{\tau}{1-\tau}}\gamma^*}$, where $\gamma^*$ is the optimized displacement. Because $\tau \to 1$, the strength of $\ket{\beta_{LO}}$ in ODTD is also high. Therefore, we can regard the one-port homodyne detection as a generalized Kennedy receiver with threshold detection and a displacement $\gamma=\sqrt{\frac{1-\tau}{\tau}}\beta_{LO}$. When $\gamma \to \infty$, we have $\beta_{LO} \to \infty$; then the generalized Kennedy receiver with threshold detection should degenerate to the one-port homodyne detection. In this case, the error probability of ODTD can be obtained using a similar approach derived in \cite{schumaker1984noise}, which results in the following Theorem \ref{theorem 3}.

\begin{theorem}\label{theorem 3}
When the displacement $\gamma$ approaches $\infty$, the threshold detection for discriminating displaced binary coherent states degenerates to the one-port homodyne detection, and the minimum error probability for \eqref{Pe_1} with an optimum threshold can be approximated by
\begin{equation}
\label{One-port_Homodyne_1}
\begin{aligned}
P_e= p_0 Q\left(\sqrt{\Lambda_1}-\frac{1}{2\sqrt{\Lambda_1}}\ln\frac{p_1}{p_0}\right)+
p_1Q\left(\sqrt{\Lambda_1}+\frac{1}{2\sqrt{\Lambda_1}}\ln\frac{p_1}{p_0}\right)
\end{aligned}
\end{equation}
\noindent where $\Lambda_1=\frac{4\tau \xi^2 N_s}{2N_t+1/\eta}$ is the signal-to-noise of one-port homodyne detection; and $Q(x)$ is the tail distribution function of the standard normal distribution.
\end{theorem}

\begin{proof}
We first consider the case of perfect detection with $\eta=1$, and denote the number of detected photons by $n_d(\pm \beta)$ for input state $\ket{\pm \beta}$. When $\gamma$ is large, the expectation of $n_d(\pm \beta)$ is $\left \langle n_d \right \rangle_{\pm}\cong \tau(\gamma^2\pm 2\xi \beta \gamma)+N_t$, and the variance of $n_d(\pm \beta)$ can be obtained as $\left \langle {\Delta n_d}^2 \right \rangle_{\pm} \cong(2N_t+1)\tau \gamma^2$ \cite{cariolaro2016quantum}.

Then we consider the imperfection of the quantum detection, i.e., $\eta < 1$, and denote the number of detected photons by $n_d'(\pm \beta)$. The imperfect detection with quantum efficiency $\eta$ is equivalent to a perfect detection after the state combined with a vacuum state using a beam splitter with transmission rate $\eta$ \cite{schumaker1984noise}. Then the expectation of $n_d'(\pm \beta)$ is $\left \langle n_d' \right \rangle_{\pm}=\eta \left \langle n_d \right \rangle_{\pm}$. The variance of $n_d'(\pm \beta)$ can be obtained as \cite{schumaker1984noise}
\begin{equation}
\label{Variance_1}
\begin{aligned}
\left \langle {\Delta n_d'}^2 \right \rangle_{\pm}&=\eta^2 \left \langle {\Delta n_d}^2 \right \rangle_{\pm} +\eta(1-\eta) \left \langle n_d \right \rangle_{\pm}\\
&\cong \eta^2\tau\gamma^2(2N_t+1/\eta).
\end{aligned}
\end{equation}

When $\gamma$ approaches $\infty$, the number of detected photons $n_d'(\pm \beta)$ for input state $\ket{\pm \beta}$ can be approximated by a Gaussian distribution with mean $\left \langle n_d' \right \rangle_{\pm}$ and variance $\left \langle {\Delta n_d'}^2 \right \rangle_{\pm}$. Using MAP detection, the optimum threshold can be obtained as
\begin{equation}
\label{Optimal_threshold_for_Gaussian_approximation}
\begin{aligned}
n_{th}=\eta \tau \gamma^2+\eta N_t -\eta\gamma\frac{2N_t+1/\eta}{4\xi \beta} \ln\frac{p_1}{p_0}.
\end{aligned}
\end{equation}

Then the error probability of distinguishing state $\ket{-\beta}$ and $\ket{\beta}$ can be approximated by
\begin{equation}
\label{One-port_Homodyne_2}
\begin{aligned}
P_e&=p_0 Q\left(\sqrt{\frac{(n_{th}-\left \langle n_d' \right \rangle_{-})^2}{\left\langle {\Delta n_d'}^2 \right\rangle_{\pm}}}\right)
+p_1 Q\left(\sqrt{\frac{(n_{th}-\left \langle n_d' \right \rangle_{+})^2}{\left\langle {\Delta n_d'}^2 \right\rangle_{\pm}}}\right)
\\
&= p_0 Q\left(\sqrt{\frac{(2\sqrt{\tau}\xi\beta-\frac{2N_t+1/\eta}{4\sqrt{\tau}\xi\beta}\ln\frac{p_1}{p_0})^2}{2 N_t +1/\eta}}\right)+ p_1 Q\left(\sqrt{\frac{(2\sqrt{\tau}\xi\beta+\frac{2N_t+1/\eta}{4\sqrt{\tau}\xi\beta}\ln\frac{p_1}{p_0})^2}{2 N_t +1/\eta}}\right).
\end{aligned}
\end{equation}

By letting $\Lambda_1=\frac{4\tau \xi^2 N_s}{2N_t+1/\eta}$, eq. \eqref{One-port_Homodyne_2} can be rewritten as \eqref{One-port_Homodyne_1}.
\end{proof}

Because the one-port homodyne detection is vulnerable to the extra noise of the LO, Yuen and Chan proposed the two-port homodyne detection \cite{yuen1983noise}. The signal-to-noise of two-port homodyne detection under thermal noise with imperfect detection can be obtained as $\Lambda_2=\frac{4N_s}{2N_t+1/\eta}$ \cite{yuen1983noise,schumaker1984noise}. Then the error probability for two-port homodyne detection is obtained by replacing $\Lambda_1$ with $\Lambda_2$ in \eqref{One-port_Homodyne_1}. Specially, when equal prior probabilities are considered, the error probability is given by
\begin{equation}
\label{Two-port_homodyne}
\begin{aligned}
P_e=Q\left(\sqrt{\frac{4N_s}{2N_t+1/\eta}}\right).
\end{aligned}
\end{equation}

Due to its robust performance against the extra noise of LO, the two-port homodyne detection is widely used in both classical and quantum communication systems. Therefore, we use the two-port homodyne detection as our reference scheme for performance comparison in Section \ref{sect:NumericalResults}.

\section{Numerical Results}
\label{sect:NumericalResults}

Unless otherwise specified, the parameters in this Section are set as follows: $p_0=p_1=0.5$, $\tau=0.99$, $\xi=0.998$, $N_t=0.01$, $\nu=0.001$, $\eta=0.72$.

\begin{figure}
\begin{center}
\subfigure[]{\includegraphics[width=0.6\textwidth]{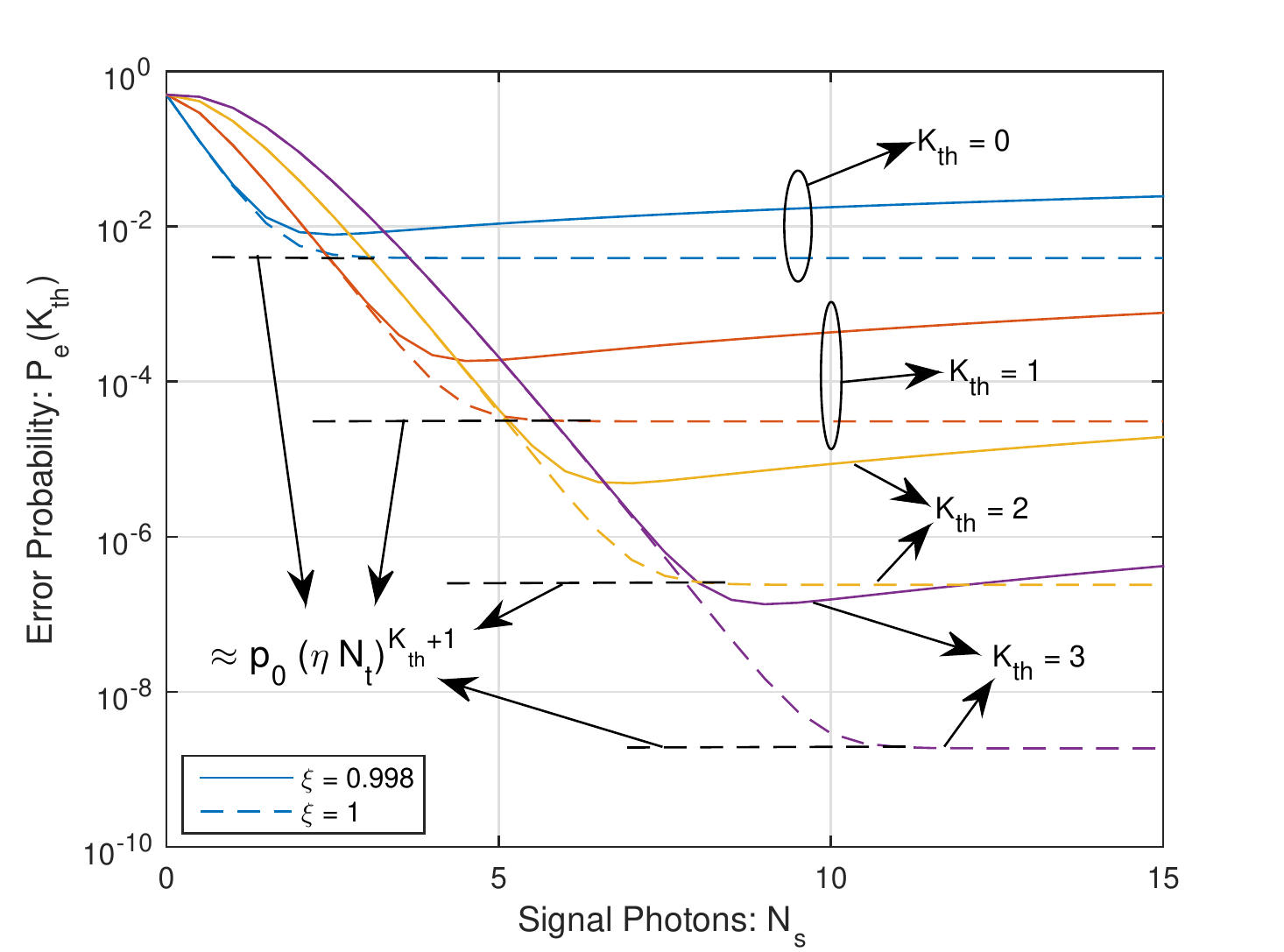}
\label{Pe_ve_Ns_different_K_xi}}
\subfigure[]{\includegraphics[width=0.6\textwidth]{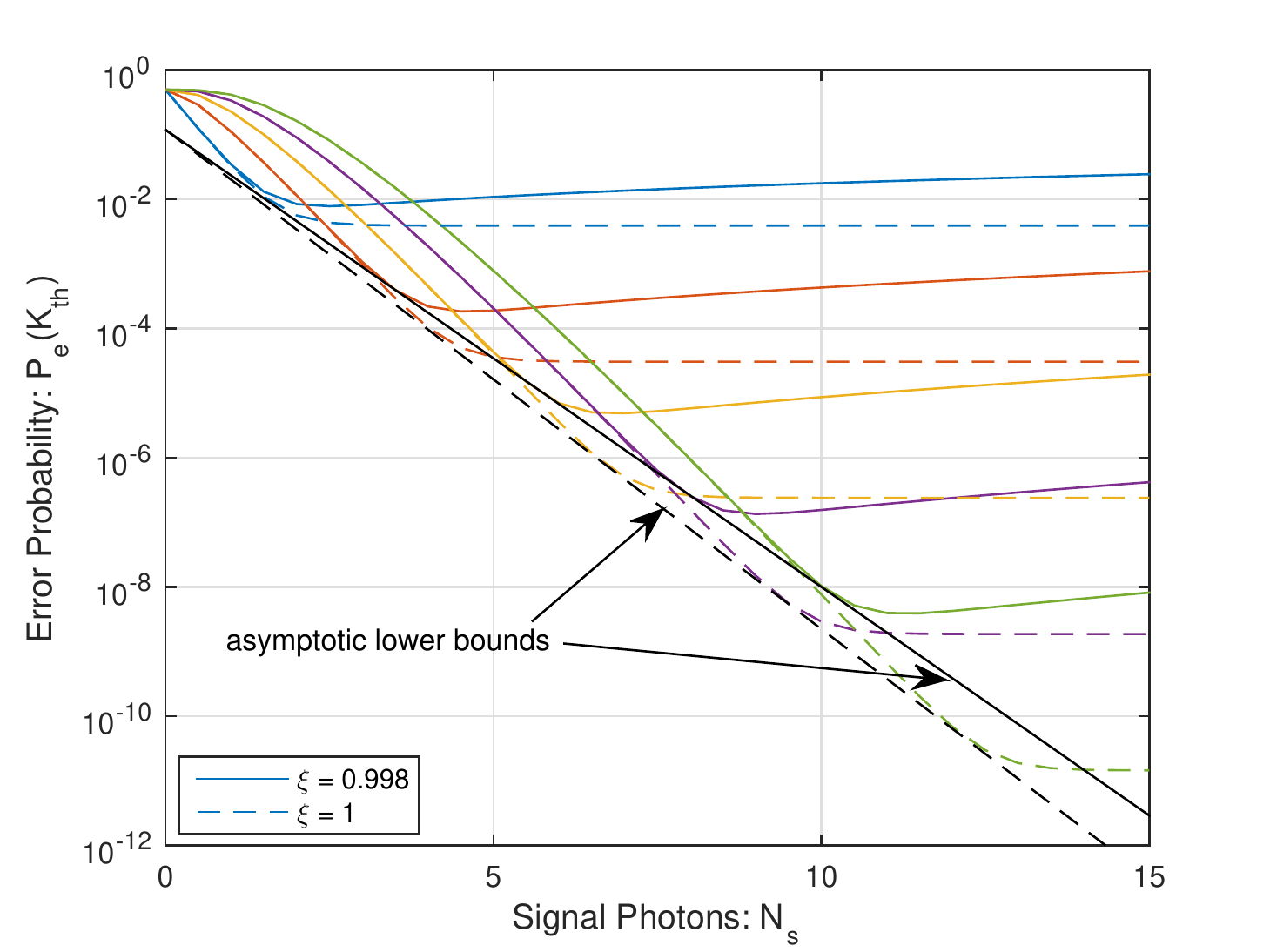}
\label{Pe_ve_Ns_different_K_asym}}
\caption{Error probabilities $P_e(K_{th})$ versus signal photons $N_s$: (a) different interference visibility $\xi$; (b) asymptotic lower bounds $P_e \approx 0.12 e^{-1.63 N_s}$ for $\xi=0.998$ in solid line and $P_e \approx 0.12 e^{-1.78 N_s}$ for $\xi=1$ in dash line}
\label{Pe_ve_Ns_different_K}
\end{center}
\end{figure}

\subsection{Properties Of TD and ODTD}

We first consider the case when $\gamma=\beta$, i.e., the generalized Kennedy receiver with threshold detection (TD) \cite{yuan2020kennedy}. The error probabilities for different threshold $K_{th}$ is shown in Fig. \ref{Pe_ve_Ns_different_K}. From Fig. \ref{Pe_ve_Ns_different_K_xi}, we can see that, for $\xi =0.998$, the error probability for a given threshold first decreases then increases when the signal power approaches infinity; while for $\xi=1$, there exist a flat lower bound around $p_0(\eta N_t)^{K_{th}+1}$ for a given threshold $K_{th}$ when the signal power approaches infinity. This indicates that the imperfect interference visibility has a great impact on the discrimination performance when the signal power is large. Besides, we can also see that the difference of error probability for a given threshold between $\xi=0.998$ and $\xi=1$ increases as the threshold $K_{th}$ increases. This indicates that the imperfect interference visibility has a great impact on the discrimination performance when the threshold is large. As we have mentioned in Section \ref{lowerbounds}, because the threshold can be optimized for a given $N_s$, the minimum error probability is the envelope of a set of curves $\{P_{e}(K_{th}), K_{th}=0,1,2,\cdots\}$ shown in Fig. \ref{Pe_ve_Ns_different_K_asym}. Then a practical lower bound should be irrelevant to the threshold $K_{th}$. In Fig. \ref{Pe_ve_Ns_different_K_asym}, we use a curve-fitting method to plot the asymptotic lower bounds $P_e \approx 0.12 e^{-1.63 N_s}$ for $\xi=0.998$ in solid line and $P_e \approx 0.12 e^{-1.78 N_s}$ for $\xi=1$ in dash line.

\begin{figure}
\begin{center}
\includegraphics[width=0.6\textwidth, draft=false]{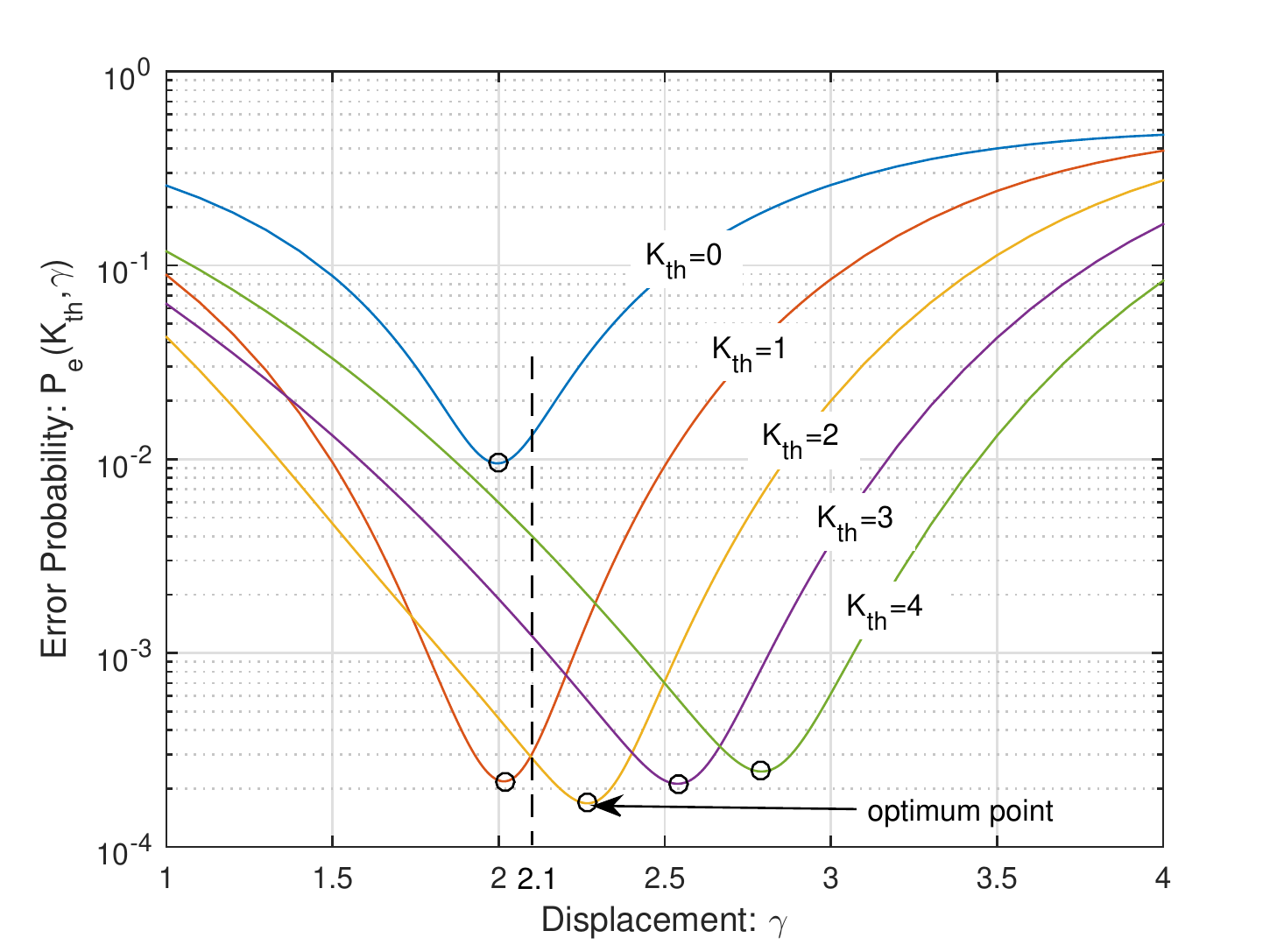}
\caption{Error probabilities $P_e(K_{th},\gamma)$ versus displacement $\gamma$ under different threshold $K_{th}$ when $\beta=2$}
\vspace{-0.4cm}
\label{Pe_ve_gamma}
\end{center}
\end{figure}

Then we consider the case when $\gamma \neq \beta $, i.e., the generalized Kennedy receiver with ODTD. Figure \ref{Pe_ve_gamma} presents the error probabilities under different threshold values $K_{th}={0, 1, 2, 3, 4}$ when the displacement $\gamma$ varies from 1 to 4 with $\beta=2$. We can see that the minimum error probability occurs when the threshold $K_{th}=2$ and the displacement $\gamma=2.25$. Suppose we use the coordinate descent algorithm to search the minimum error probability. When the displacement $\gamma$ falls into $(0,2.1)$, following the coordinate descent, we can obtain the optimum threshold in this case as $K_{th}=1$. Then we optimize the displacement for $K_{th}=1$ and obtain $\gamma=2.02$. Because this displacement still falls in $(0,2.1)$, the coordinate descent will get stuck at this locally optimal point $(K_{th}=2,\gamma=2.02)$. Therefore, the coordinate descent can fail in searching the minimum error probability. On the other hand, the minimum error probability for a given threshold, i.e., $P_e(K_{th}, \gamma^*(K_{th}))$, decreases when the threshold $K_{th}$ increases from 0 to 2 and increases when $K_{th}$ increases from 2 to 4. This indicates that $P_e(K_{th}, \gamma^*(K_{th}))$ is a discrete convex function of $K_{th}$. Therefore, we can use the heuristic greedy search algorithm proposed in Section \ref{algorithm} to search the minimum error probability.

\begin{figure}
\begin{center}
\subfigure[]{\includegraphics[width=0.6\textwidth]{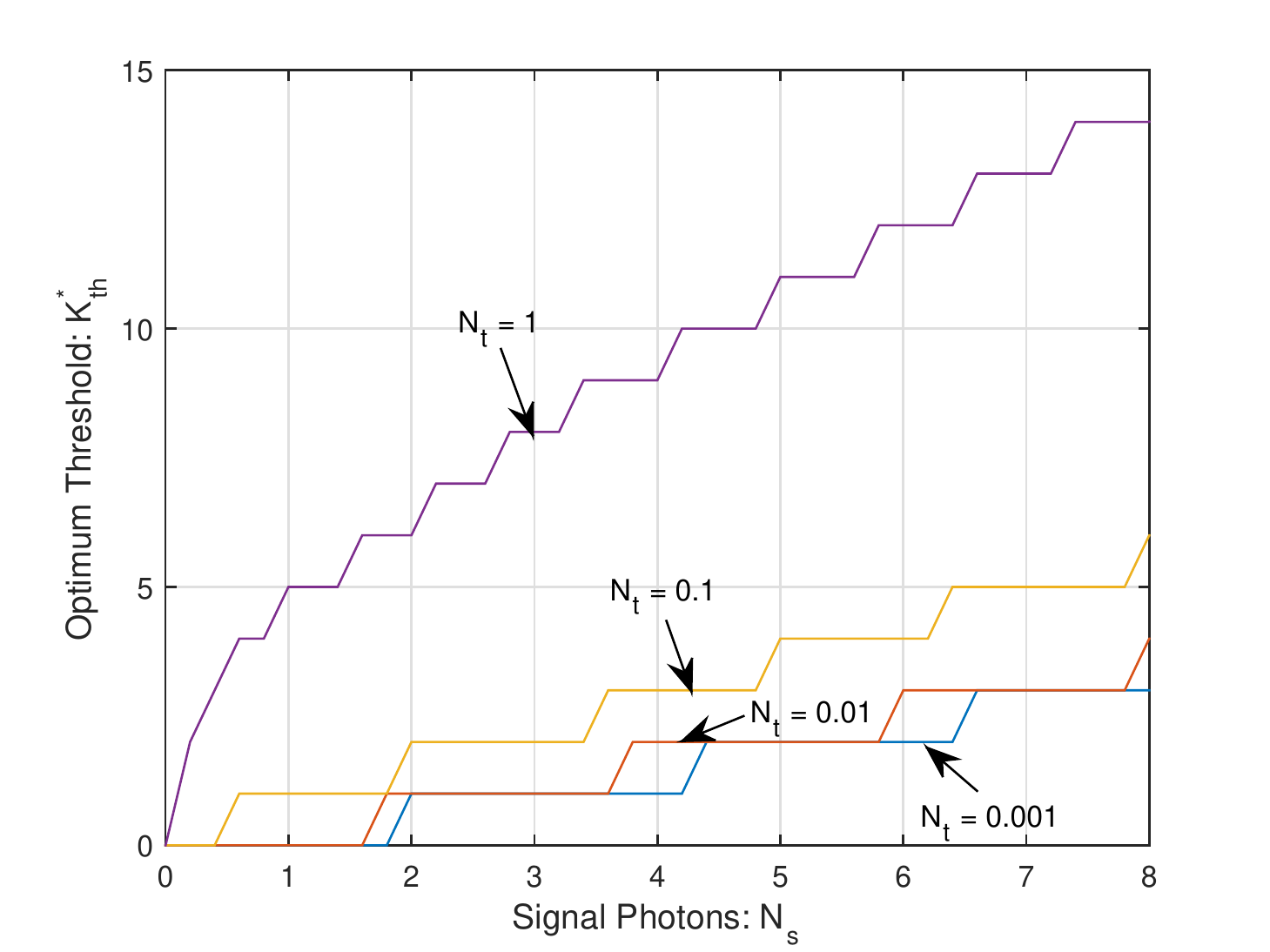}
\label{Optimal threshold}}
\subfigure[]{\includegraphics[width=0.6\textwidth]{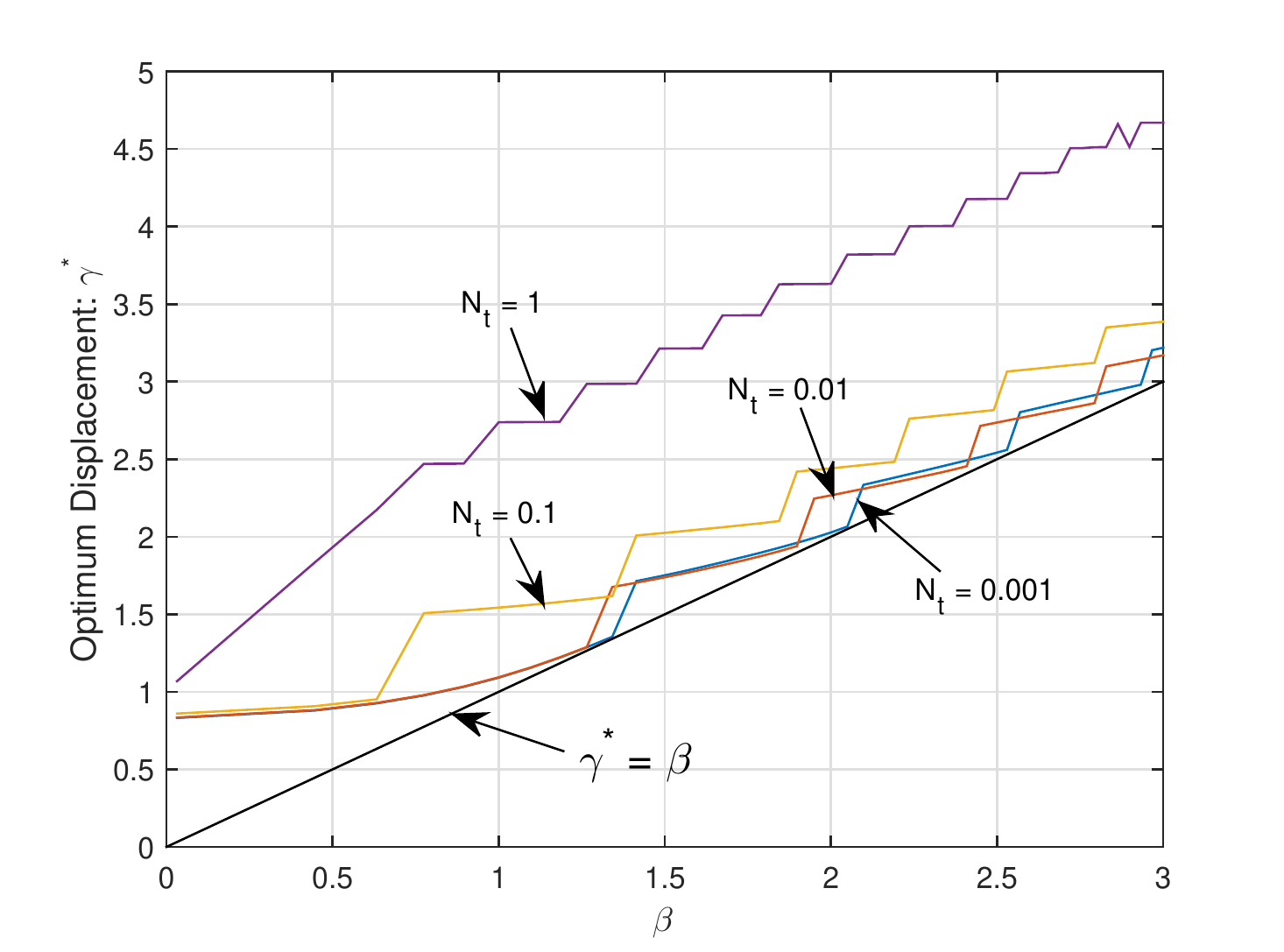}
\label{optimal displacement}}
\caption{Optimum threshold and optimum displacement for ODTD under different thermal photons (a) optimum threshold $K_{th}^*$; (b) optimum displacement $\gamma^*$}
\end{center}
\end{figure}

The obtained optimum threshold and optimum displacement for ODTD using the heuristic greedy search algorithm are shown in Fig. \ref{Optimal threshold} and Fig. \ref{optimal displacement}, respectively. From Fig. \ref{Optimal threshold} we can see that the optimum threshold $K_{th}^*$ increases as either the signal photons $N_s$ or thermal photons $N_t$ increases. From Fig. \ref{optimal displacement} we can see that the optimum threshold $\gamma^*$ increases as the signal $\beta$ increases, and the optimum threshold $\gamma^*$ is always greater than the signal $\beta$. Besides, the gap between $\gamma^*$ and $\beta$ increases as the thermal photons $N_t$ increases.

\begin{figure}
\begin{center}
\includegraphics[width=0.6\textwidth, draft=false]{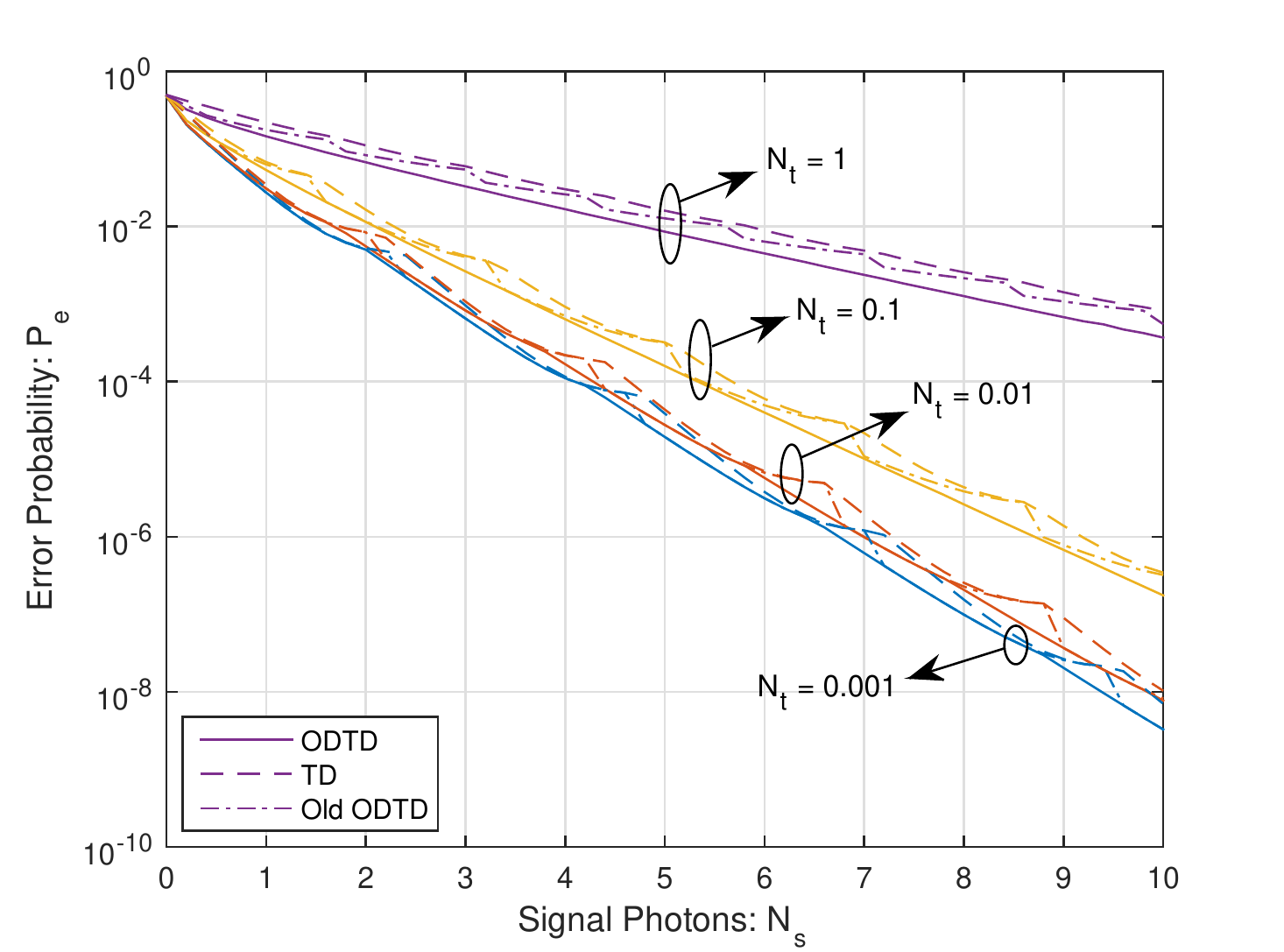}
\caption{Error probabilities $P_e$ versus signal photons $N_s$ under different thermal photons $N_t$}
\vspace{-0.4cm}
\label{ODTD_vs_TD_different_Nt}
\end{center}
\end{figure}

The minimum error probability comparison between the generalized Kennedy receiver with TD ($\gamma=\beta$) and ODTD ($\gamma \neq \beta$) under different thermal photons is shown in Fig. \ref{ODTD_vs_TD_different_Nt}. We also plot the error probabilities (old ODTD) obtained by our previous work \cite{yuan2020kennedy}, where the displacement is optimized after the threshold optimization. We can see that the error probabilities of ODTD are lower than those of both TD and old ODTD. Besides, the error probability curves of ODTD are much smoother than those of TD and old ODTD, especially when $N_t$ is small. Because the error probabilities of ODTD are always lower than those of TD and old ODTD, we will focus on the performance comparison between the ODTD and the homodyne detection in the following.

\subsection{Comparison Between ODTD and Homodyne Detection}

\begin{figure}
\begin{center}
\subfigure[]{\includegraphics[width=0.6\textwidth]{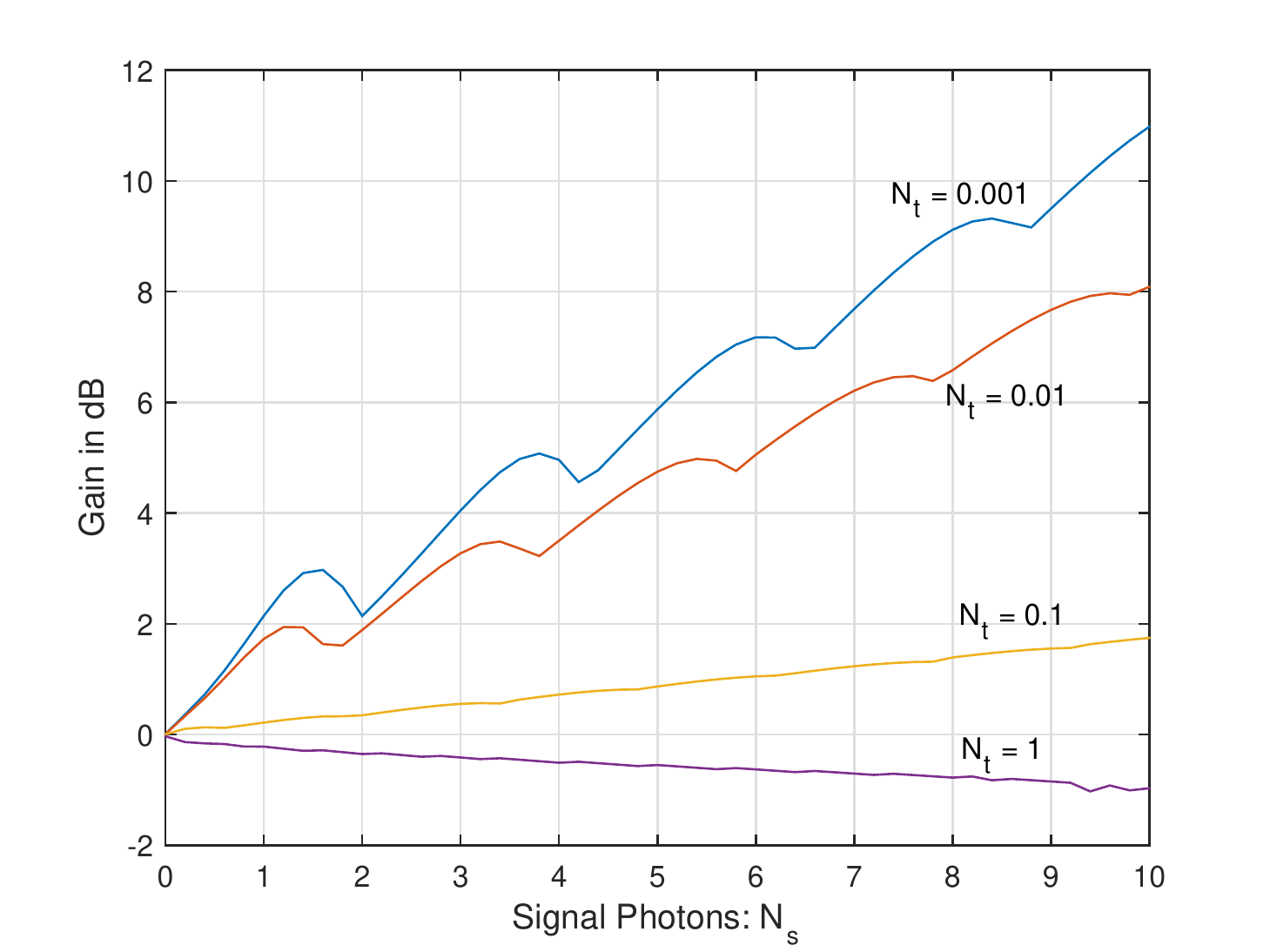}
\label{Gain_in_different_Nt}}
\subfigure[]{\includegraphics[width=0.6\textwidth]{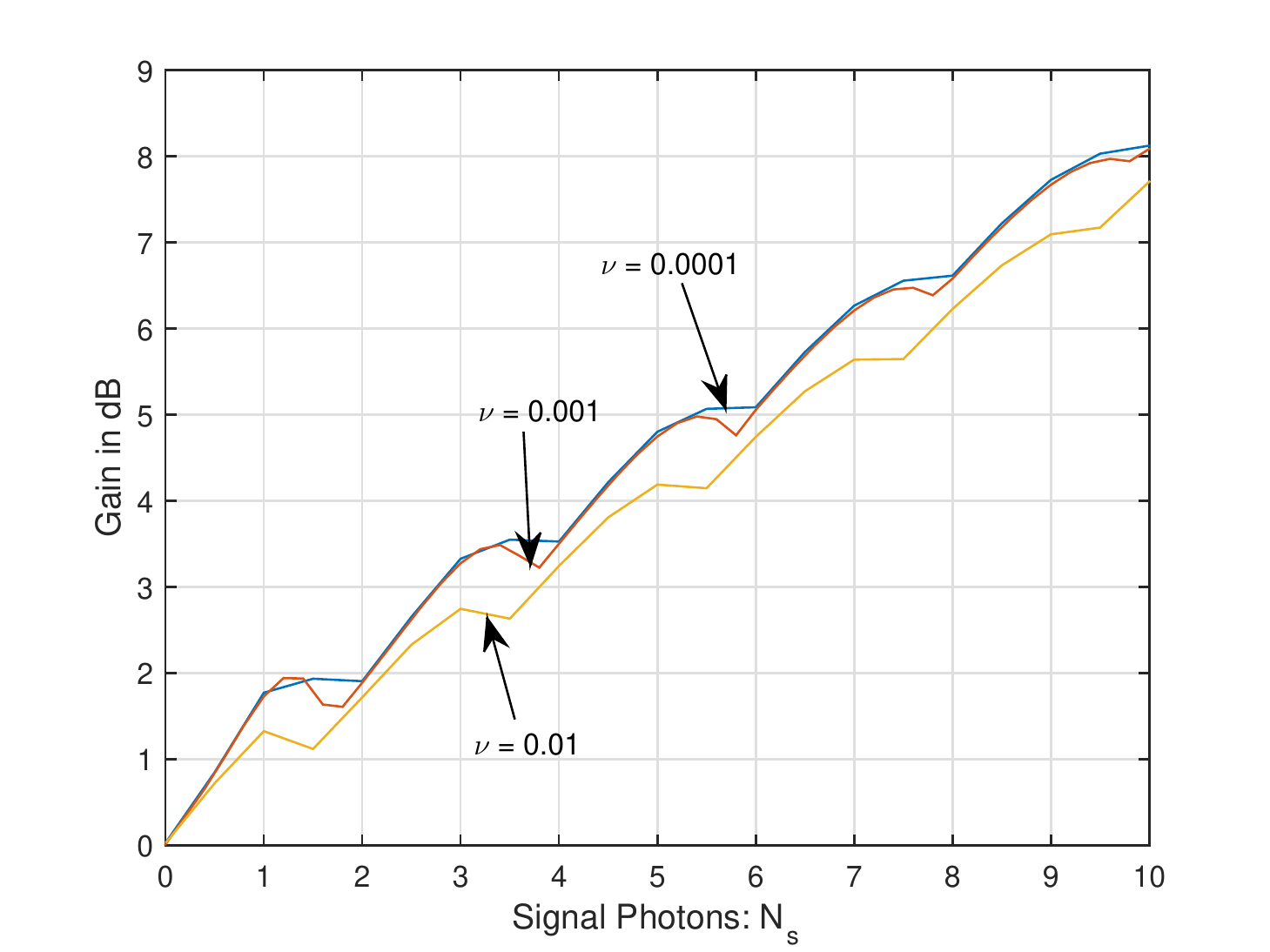}
\label{Gain_in_different_nu}}
\caption{Gain in dB of ODTD over homodyne detection (a) different thermal photons $N_t$; (b) different dark counts $\nu$}
\label{InfluencesOfNoises}
\end{center}
\end{figure}

For performance comparison, we define the performance gain in dB of ODTD over homodyne detection as $10\times \lg \frac{P_{e,homodyne}}{P_{e,ODTD}}$. Figure \ref{InfluencesOfNoises} shows the influences of noises on the performance gain, where Fig. \ref{Gain_in_different_Nt} and Fig. \ref{Gain_in_different_nu} show the gains under different thermal photons and different dark counts, respectively. The thermal photons can vary in a large range due to the variation of working temperature or potential Gaussian attacks on the channel, while the dark counts are much stabler and smaller than the thermal photons during the working period. Therefore, we set the range of thermal photons as from $0.001$ to $1$ and the range of dark counts as from $0.0001$ to $0.01$. Comparing Fig. \ref{Gain_in_different_Nt} and Fig. \ref{Gain_in_different_nu}, we can see that the influence of the thermal photons is much greater than that of the dark counts, which is as expected. Besides, from Fig. \ref{Gain_in_different_Nt}, we can see that the performance of ODTD can surpass the SQL given by the homodyne detection when $N_t \leq 0.1$.

\begin{figure}
\begin{center}
\subfigure[]{\includegraphics[width=0.6\textwidth]{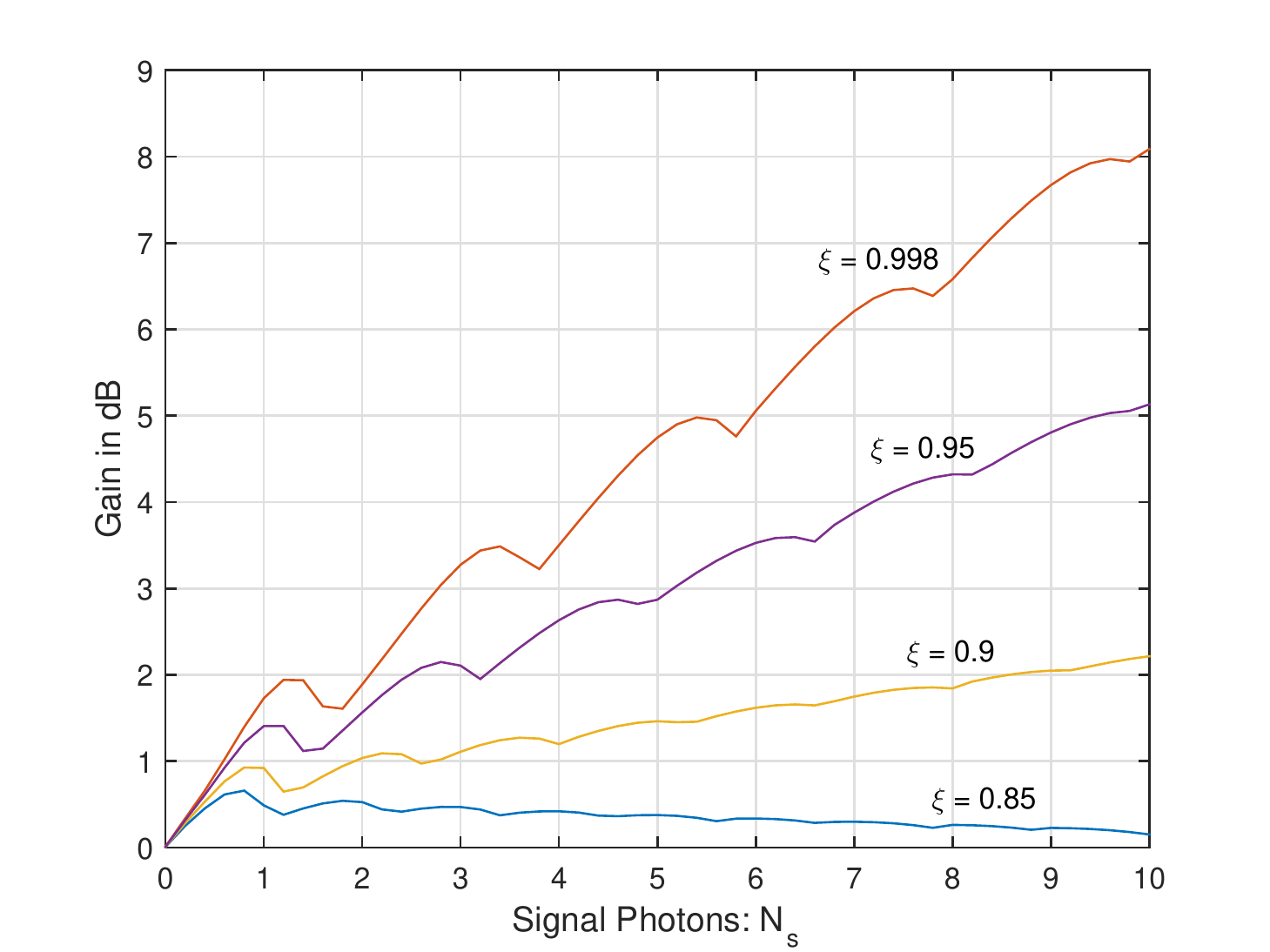}
\label{Gain_in_different_xi}}
\subfigure[]{\includegraphics[width=0.6\textwidth]{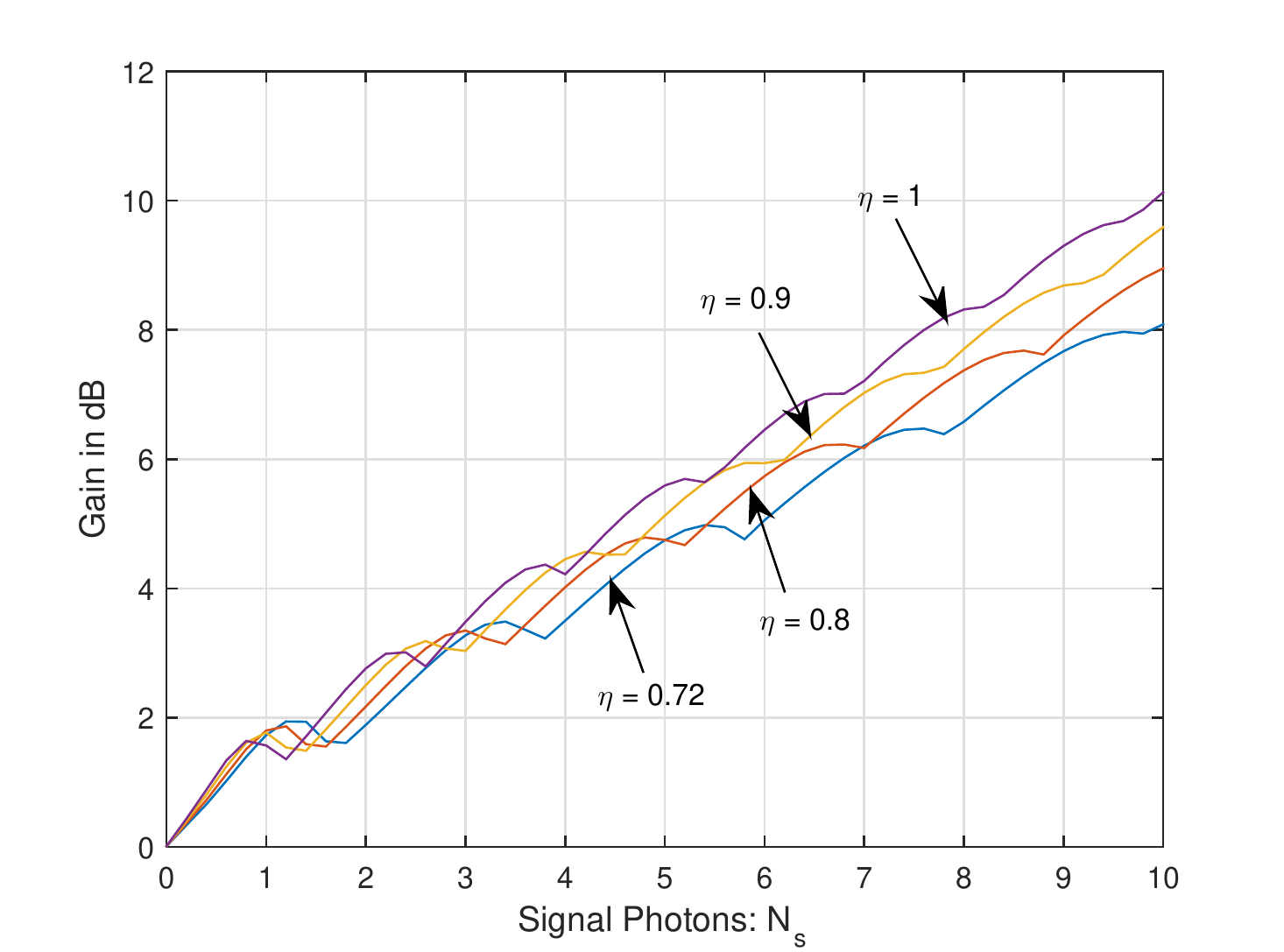}
\label{Gain_in_different_eta}}
\caption{Gain in dB of ODTD over homodyne detection (a) different interference visibility $\xi$; (b) different quantum efficiency $\eta$}
\label{InfluencesOfImperfections}
\end{center}
\end{figure}

Figure \ref{InfluencesOfImperfections} shows the influences of device imperfections on the performance gain, where Fig. \ref{Gain_in_different_xi} and Fig. \ref{Gain_in_different_eta} show the gains under different interference visibility and different quantum efficiency, respectively. From Fig. \ref{Gain_in_different_xi}, we can see that the performance of ODTD can surpass the SQL even if the interference visibility $\xi=0.85$ when $N_s \leq 10$. A lower interference visibility than $0.85$ can destroy the advantage of ODTD, especially in large signal powers. Comparing Fig. \ref{Gain_in_different_xi} and Fig. \ref{Gain_in_different_eta}, we can see that the influence of the imperfect interference visibility is much greater than that of the quantum efficiency. This is because the effect of a quantum efficiency $\eta$ is equivalent to a beam splitter with transmission rate $\eta$, and therefore it can be regarded as signal power degradation for both ODTD and homodyne detection.

\begin{figure}
\begin{center}
\includegraphics[width=0.6\textwidth, draft=false]{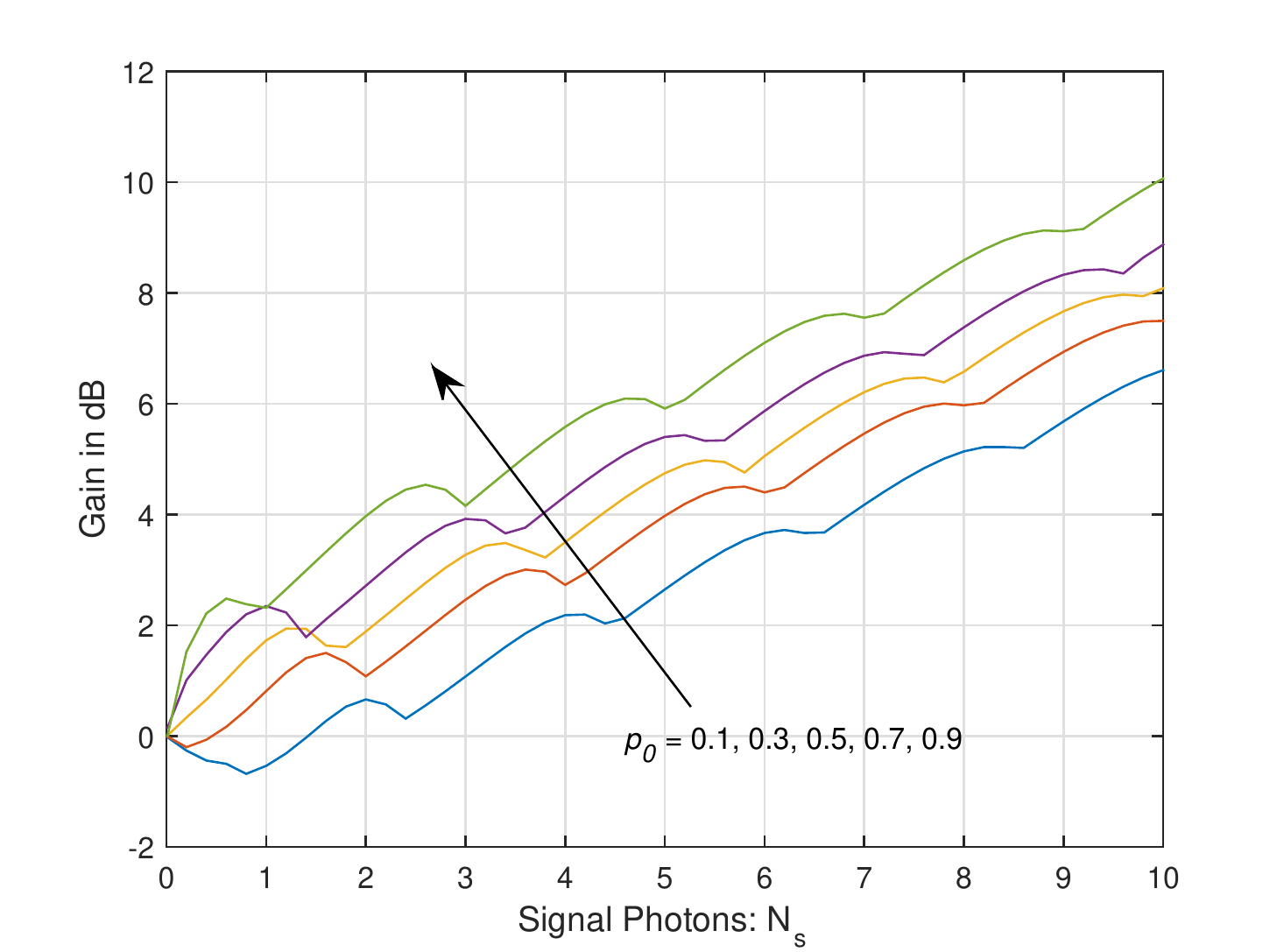}
\caption{Gain in dB of ODTD over homodyne detection in different $p_0$}
\vspace{-0.4cm}
\label{Gain_in_different_p_0}
\end{center}
\end{figure}

At last, we check the performance gain of ODTD over homodyne detection in different prior probabilities. As shown in Fig. \ref{Gain_in_different_p_0}, the gain increases as the prior probability $p_0$ increases. This indicates that the ODTD has more advantage over homodyne detection for a large $p_0$ compared with a small $p_0$.

\section{Conclusion}
\label{sect:Conclusion}

We analytically studied the generalized Kennedy receiver with ODTD in a realistic situation in the presence of both thermal noise and dark count noise using imperfect devices. We first proved that the MAP detection for generalized Kennedy receiver in this realistic situation is equivalent to a threshold detection. Then we analyzed the properties of the optimum threshold and the optimum displacement for ODTD and proposed a heuristic greedy search algorithm to obtain them. We proved that the ODTD degenerates to the Kennedy receiver with threshold detection when the signal power is large. We also clarified the connection between the generalized Kennedy receiver with threshold detection and the one-port homodyne detection. We found that the proposed  heuristic greedy search algorithm can obtain a lower and smoother error probability than that of previous works.

\bibliographystyle{IEEEtran}
\bibliography{refabrv}

\begin{thebibliography}{10}
\providecommand{\url}[1]{#1}
\csname url@samestyle\endcsname
\providecommand{\newblock}{\relax}
\providecommand{\bibinfo}[2]{#2}
\providecommand{\BIBentrySTDinterwordspacing}{\spaceskip=0pt\relax}
\providecommand{\BIBentryALTinterwordstretchfactor}{4}
\providecommand{\BIBentryALTinterwordspacing}{\spaceskip=\fontdimen2\font plus
\BIBentryALTinterwordstretchfactor\fontdimen3\font minus
  \fontdimen4\font\relax}
\providecommand{\BIBforeignlanguage}[2]{{%
\expandafter\ifx\csname l@#1\endcsname\relax
\typeout{** WARNING: IEEEtran.bst: No hyphenation pattern has been}%
\typeout{** loaded for the language `#1'. Using the pattern for}%
\typeout{** the default language instead.}%
\else
\language=\csname l@#1\endcsname
\fi
#2}}
\providecommand{\BIBdecl}{\relax}
\BIBdecl

\bibitem{osaki1996derivation}
M.~Osaki, M.~Ban, and O.~Hirota, ``Derivation and physical interpretation of
  the optimum detection operators for coherent-state signals,'' \emph{Phys.
  Rev. A}, vol.~54, no.~2, p. 1691, Aug. 1996.

\bibitem{vilnrotter2001quantum}
V.~Vilnrotter and C.~Lau, ``Quantum detection theory for the free-space
  channel,'' \emph{The IPN Report 42-146, April--June 2001}, pp. 1--34, 2001.

\bibitem{lance2005no}
A.~M. Lance, T.~Symul, V.~Sharma, C.~Weedbrook, T.~C. Ralph, and P.~K. Lam,
  ``No-switching quantum key distribution using broadband modulated coherent
  light,'' \emph{Phys. Rev. Lett}, vol.~95, no.~18, p. 180503, Oct. 2005.

\bibitem{weedbrook2012gaussian}
C.~Weedbrook, S.~Pirandola, R.~Garc{\'\i}a-Patr{\'o}n, N.~J. Cerf, T.~C. Ralph,
  J.~H. Shapiro, and S.~Lloyd, ``Gaussian quantum information,'' \emph{Rev.
  Mod. Phys}, vol.~84, no.~2, p. 621, May. 2012.

\bibitem{yuan2018free}
R.~Yuan and J.~Cheng, ``Free-space optical quantum bpsk communications in
  turbulent channels,'' in \emph{2018 IEEE Globecom Workshops (GC
  Wkshps)}.\hskip 1em plus 0.5em minus 0.4em\relax IEEE, Dec. 2018, pp. 1--6.

\bibitem{grosshans2003quantum}
F.~Grosshans, G.~Van~Assche, J.~Wenger, R.~Brouri, N.~J. Cerf, and P.~Grangier,
  ``Quantum key distribution using gaussian-modulated coherent states,''
  \emph{Nature}, vol. 421, no. 6920, p. 238, Jan. 2003.

\bibitem{lorenz2004continuous}
S.~Lorenz, N.~Korolkova, and G.~Leuchs, ``Continuous-variable quantum key
  distribution using polarization encoding and post selection,'' \emph{Appl.
  Phys. B}, vol.~79, no.~3, pp. 273--277, Aug. 2004.

\bibitem{bonato2009feasibility}
C.~Bonato, A.~Tomaello, V.~Da~Deppo, G.~Naletto, and P.~Villoresi,
  ``Feasibility of satellite quantum key distribution,'' \emph{New J. Phys},
  vol.~11, no.~4, p. 045017, Apr. 2009.

\bibitem{yuan2020closed}
R.~Yuan and J.~Cheng, ``Closed-form density matrices of free-space optical
  quantum communications in turbulent channels,'' \emph{IEEE Commun. Lett.},
  Feb. 2020.

\bibitem{yin2016measurement}
H.-L. Yin, T.-Y. Chen, Z.-W. Yu, H.~Liu, L.-X. You, Y.-H. Zhou, S.-J. Chen,
  Y.~Mao, M.-Q. Huang, W.-J. Zhang \emph{et~al.},
  ``Measurement-device-independent quantum key distribution over a 404 km
  optical fiber,'' \emph{Phys. Rev. Lett}, vol. 117, no.~19, p. 190501, Nov.
  2016.

\bibitem{ghorai2019asymptotic}
S.~Ghorai, P.~Grangier, E.~Diamanti, and A.~Leverrier, ``Asymptotic security of
  continuous-variable quantum key distribution with a discrete modulation,''
  \emph{Physical Review X}, vol.~9, no.~2, p. 021059, June 2019.

\bibitem{yuan2020freespace}
R.~Yuan and J.~Cheng, ``Free-space optical quantum communications in turbulent
  channels with receiver dversity (accepted),'' \emph{IEEE Trans. Commun},
  2020. doi: 10.1109/TCOMM.2020.2997398.

\bibitem{takeoka2008discrimination}
M.~Takeoka and M.~Sasaki, ``Discrimination of the binary coherent signal:
  Gaussian-operation limit and simple non-gaussian near-optimal receivers,''
  \emph{Phys. Rev. A}, vol.~78, no.~2, p. 022320, Aug. 2008.

\bibitem{helstrom1969quantum}
C.~W. Helstrom, ``Quantum detection and estimation theory,'' \emph{Journal of
  Statistical Physics}, vol.~1, no.~2, pp. 231--252, Jun. 1969.

\bibitem{helstrom1970quantum}
C.~W. Helstrom, J.~W. Liu, and J.~P. Gordon, ``Quantum-mechanical communication
  theory,'' \emph{Proc. IEEE}, vol.~58, no.~10, pp. 1578--1598, Oct. 1970.

\bibitem{dolinar1973optimum}
S.~J. Dolinar, ``An optimum receiver for the binary coherent state quantum
  channel,'' \emph{Quarterly Progress Report}, vol. 111, pp. 115--120, Oct.
  1973.

\bibitem{cook2007optical}
R.~L. Cook, P.~J. Martin, and J.~M. Geremia, ``Optical coherent state
  discrimination using a closed-loop quantum measurement,'' \emph{Nature}, vol.
  446, no. 7137, p. 774, Apr. 2007.

\bibitem{kennedy1973near}
R.~S. Kennedy, ``A near-optimum receiver for the binary coherent state quantum
  channel,'' \emph{Quarterly Progress Report}, vol. 108, pp. 219--225, Jan.
  1973.

\bibitem{vilnrotter1984generalization}
V.~Vilnrotter and E.~Rodemich, ``A generalization of the near-optimum binary
  coherent state receiver concept (corresp.),'' \emph{IEEE Trans. Inf. Theory},
  vol.~30, no.~2, pp. 446--450, Mar 1984.

\bibitem{bondurant1993near}
R.~S. Bondurant, ``Near-quantum optimum receivers for the phase-quadrature
  coherent-state channel,'' \emph{Opt. Lett}, vol.~18, no.~22, pp. 1896--1898,
  Nov. 1993.

\bibitem{geremia2004distinguishing}
J.~Geremia, ``Distinguishing between optical coherent states with imperfect
  detection,'' \emph{Phys. Rev. A}, vol.~70, no.~6, p. 062303, Dec. 2004.

\bibitem{wittmann2008demonstration}
C.~Wittmann, M.~Takeoka, K.~N. Cassemiro, M.~Sasaki, G.~Leuchs, and U.~L.
  Andersen, ``Demonstration of near-optimal discrimination of optical coherent
  states,'' \emph{Phys. Rev. Lett}, vol. 101, no.~21, p. 210501, Nov. 2008.

\bibitem{wittmann2010discrimination}
C.~Wittmann, U.~L. Andersen, and G.~Leuchs, ``Discrimination of optical
  coherent states using a photon number resolving detector,'' \emph{J. Mod.
  Opt.}, vol.~57, no.~3, pp. 213--217, Feb. 2010.

\bibitem{wittmann2010demonstration}
C.~Wittmann, U.~L. Andersen, M.~Takeoka, D.~Sych, and G.~Leuchs,
  ``Demonstration of coherent-state discrimination using a
  displacement-controlled photon-number-resolving detector,'' \emph{Phys. Rev.
  Lett.}, vol. 104, no.~10, p. 100505, Mar. 2010.

\bibitem{becerra2015photon}
F.~Becerra, J.~Fan, and A.~Migdall, ``Photon number resolution enables quantum
  receiver for realistic coherent optical communications,'' \emph{Nat.
  Photonics}, vol.~9, no.~1, p.~48, Jan. 2015.

\bibitem{dimario2018robust}
M.~DiMario and F.~Becerra, ``Robust measurement for the discrimination of
  binary coherent states,'' \emph{Phys. Rev. Lett.}, vol. 121, no.~2, p.
  023603, July 2018.

\bibitem{dimario2019optimized}
M.~DiMario, L.~Kunz, K.~Banaszek, and F.~Becerra, ``Optimized communication
  strategies with binary coherent states over phase noise channels,'' \emph{npj
  Quantum Inf.}, vol.~5, no.~1, pp. 1--7, July 2019.

\bibitem{shcherbatenko2020sub}
M.~Shcherbatenko, M.~Elezov, G.~Goltsman, and D.~Sych, ``Sub-shot-noise-limited
  fiber-optic quantum receiver,'' \emph{Phys. Rev. A}, vol. 101, no.~3, p.
  032306, Mar. 2020.

\bibitem{yuan2020kennedy}
R.~Yuan, M.~Zhao, S.~Han, and J.~Cheng, ``Kennedy receiver using threshold
  detection and optimized displacement under thermal noise,'' \emph{IEEE
  Commun. Lett.}, Mar 2020.

\bibitem{xu2011impact}
F.~Xu, M.-A. Khalighi, and S.~Bourennane, ``Impact of different noise sources
  on the performance of {PIN}-and {APD}-based {FSO} receivers,'' in
  \emph{Proceedings of the 11th International Conference on
  Telecommunications}.\hskip 1em plus 0.5em minus 0.4em\relax IEEE, June 2011,
  pp. 211--218.

\bibitem{mandel1995optical}
L.~Mandel and E.~Wolf, \emph{Optical coherence and quantum optics}.\hskip 1em
  plus 0.5em minus 0.4em\relax Cambridge university press, 1995.

\bibitem{wang2007quantum}
X.-B. Wang, T.~Hiroshima, A.~Tomita, and M.~Hayashi, ``Quantum information with
  gaussian states,'' \emph{Phys. Rep}, vol. 448, no. 1-4, pp. 1--111, Aug.
  2007.

\bibitem{glauber1963quantum}
R.~J. Glauber, ``The quantum theory of optical coherence,'' \emph{Phys. Rev},
  vol. 130, no.~6, p. 2529, Jun. 1963.

\bibitem{glauber1963coherent}
------, ``Coherent and incoherent states of the radiation field,'' \emph{Phys.
  Rev}, vol. 131, no.~6, p. 2766, Sep. 1963.

\bibitem{murota2001relationship}
K.~Murota and A.~Shioura, ``Relationship of m-/l-convex functions with discrete
  convex functions by miller and favati--tardella,'' \emph{Discrete Applied
  Mathematics}, vol. 115, no. 1-3, pp. 151--176, Nov. 2001.

\bibitem{yuen1983noise}
H.~P. Yuen and V.~W. Chan, ``Noise in homodyne and heterodyne detection,''
  \emph{Opt. Lett}, vol.~8, no.~3, pp. 177--179, Mar. 1983.

\bibitem{schumaker1984noise}
B.~L. Schumaker, ``Noise in homodyne detection,'' \emph{Opt. Lett}, vol.~9,
  no.~5, pp. 189--191, May 1984.

\bibitem{cariolaro2016quantum}
G.~Cariolaro, \emph{Quantum Communications}.\hskip 1em plus 0.5em minus
  0.4em\relax Springer, 2016.

\bibitem{skovgaard1954inequalities}
H.~Skovgaard, ``On inequalities of the tur{\'a}n type,'' \emph{Mathematica
  Scandinavica}, vol.~2, no.~1, pp. 65--73, Aug. 1954.

\end{thebibliography}

\appendices
\section{Proof of the decreasing property for $f(x)$}\label{Appendix A}

\begin{lemma}\label{lemma_1}
When $x\leq 0$ and $m>n$, the inequality $L_{m-1}^{\alpha}(x)L_{n}^{\alpha}(x)>L_m^{\alpha}(x)L_{n-1}^{\alpha}(x)$ always holds, where $\alpha \in \mathbb{R}$ and $\mathbb{R}$ is the set of real numbers.
\end{lemma}

\begin{proof}
According to the Laguerre inequalities of the Turán type \cite{skovgaard1954inequalities}, we have
\begin{equation}\label{Turan inequalities}
\begin{aligned}
L_{m-1}^{\alpha}(x)L_{m-1}^{\alpha}(x)&>L_{m}^{\alpha}(x)L_{m-2}^{\alpha}(x)\\
L_{m-2}^{\alpha}(x)L_{m-2}^{\alpha}(x)&>L_{m-1}^{\alpha}(x)L_{m-3}^{\alpha}(x)\\
&\cdots\\
L_{n}^{\alpha}(x)L_{n}^{\alpha}(x)&>L_{n+1}^{\alpha}(x)L_{n-1}^{\alpha}(x).
\end{aligned}
\end{equation}
For $x\leq 0$, the inequality $L_{n}^{\alpha}(x) > 0$ always holds. Therefore, we can multiply all the terms at the left side and the terms at the right side of \eqref{Turan inequalities} without changing the direction of the inequality sign. By canceling the same terms in both sides, we have $L_{m-1}^{\alpha}(x)L_{n}^{\alpha}(x)>L_{m}^{\alpha}(x)L_{n-1}^{\alpha}(x)$.

\end{proof}

\begin{lemma}\label{lemma_2}
When $x\leq 0$ and $m>n$, the function $f(x)\triangleq \frac{L_{m}(x)}{L_n(x)}$ is a decreasing function of $x$.
\end{lemma}
\begin{proof}
The derivative of $f(x)$ can be obtained as
\begin{equation}
\begin{aligned}
\frac{\mathrm{d}f(x)}{\mathrm{d}x}&=\frac{-L_{m-1}^1(x)L_n(x)+L_m(x)L_{n-1}^1(x)}{L_n(x)^2}\\
&=\frac{-L_{m-1}^1(x)[L_n^1(x)-L_{n-1}^1(x)]+[L_m^1(x)-L_{m-1}^1(x)]L_{n-1}^1(x)}{L_n(x)^2}\\
&=\frac{-L_{m-1}^1(x)L_n^1(x)+L_m^1(x)L_{n-1}^1(x)}{L_n(x)^2}\\
&<0
\end{aligned}
\end{equation}
\noindent where in the last step we have used Lemma \ref{lemma_1}. Therefore, $f(x)$ is a decreasing function of $x$.
\end{proof}
\balance
\end{document}